\documentclass[a4paper]{article}

\usepackage[utf8]{inputenc}
\usepackage[english]{babel}

\usepackage{fullpage} 

\usepackage[textsize=scriptsize]{todonotes}

\usepackage{authblk}

\usepackage[shortlabels]{enumitem}

\usepackage{amsmath}
\usepackage{amssymb}
\usepackage{bbm} 
\usepackage{mathrsfs}

\usepackage{url}
\usepackage{hyperref}
\usepackage{bookmark}
\hypersetup{
    colorlinks = true,
    citecolor = blue,
    linkcolor = blue,
    bookmarksnumbered
}

\usepackage{graphicx}
\usepackage[width=.9\textwidth]{caption}
\usepackage{subcaption}
\usepackage{float}

\usepackage{tikz}
\usetikzlibrary{calc,fit,shapes}

\usepackage{cite}

\usepackage[ruled, vlined, linesnumbered]{algorithm2e}
\SetKwFor{RepTimes}{repeat}{times}

\usepackage[noabbrev]{cleveref}

\usepackage{xspace}

\usepackage{centernot}

\usepackage[misc]{ifsym}

\usepackage[title]{appendix}

\usepackage{amsthm}
\theoremstyle{plain}
    \newtheorem{theorem}{Theorem}
    \newtheorem{lemma}{Lemma}
    \newtheorem{proposition}{Proposition}
    \newtheorem{claim}{Claim}
    
\theoremstyle{definition}
    \newtheorem{definition}{Definition}

\crefname{definition}{definition}{definitions}
\crefname{theorem}{theorem}{theorems}
\crefname{corollary}{corollary}{corollaries}
\crefname{lemma}{lemma}{lemmas}
\crefname{proposition}{proposition}{propositions}
\crefname{claim}{claim}{claims}
\crefname{remark}{remark}{remarks}
\usepackage{thm-restate}

\newcommand{\R}{\mathbb{R}}
\newcommand{\Z}{\mathbb{Z}}

\newcommand{\cC}{\mathcal{C}}
\newcommand{\cX}{\mathcal{X}}
\newcommand{\cP}{\mathcal{P}}
\newcommand{\cE}{\mathcal{E}}
\newcommand{\cG}{\mathcal{G}}

\newcommand{\Par}[1]{\left( #1 \right)}

\newcommand{\Abs}[1]{\left| #1 \right|}
\newcommand{\CBra}[1]{\left\{ #1 \right\}}

\newcommand{\define}[1]{\emph{#1}}

\newcommand{\graph}{G}
\newcommand{\vertexSet}{V}
\newcommand{\edgeSet}{E}
\newcommand{\graphVE}{\graph=(\vertexSet,\edgeSet)}

\newcommand{\shortEdge}[2]{#1#2}
\newcommand{\symdif}{\triangle}
\newcommand{\matching}{M}

\newcommand{\capacity}{c}
\newcommand{\weight}{w}
\newcommand{\graphwc}{(\graph,\weight,\capacity)}
\newcommand{\graphwunit}{(\graph,\weight,1)}
\newcommand{\graphunitc}{(\graph,1,\capacity)}

\newcommand{\gwcm}{[\graphwc,\matching]}

\newcommand{\oddcyclesx}[1]{\mathscr{C}_{#1}}
\newcommand{\matchedx}[1]{\mathscr{M}_{#1}}

\newcommand{\nufcG}{\nu_f^\capacity(\graph)}
\newcommand{\nucG}{\nu^\capacity(\graph)}

\newcommand{\taufcG}{\tau_f^\capacity(\graph)}

\newcommand{\polyhedron}{\mathcal{P}}
\newcommand{\Pfcm}{\mathcal{P}_{\text{FCM}}}

\newcommand{\Pfpcm}{\mathcal{P}_{\text{FPCM}}}

\newcommand{\Po}{\text{P}\xspace}
\newcommand{\NP}{\text{NP}\xspace}

\newcommand\restr[2]{\left.#1\right|_{#2}}

\def\keywordname{{\bf Keywords:}}
\providecommand{\keywords}[1]{\def\and{{\textperiodcentered} }
\par\addvspace\baselineskip
\noindent\keywordname\enspace\ignorespaces#1}

\def\fundingname{{\bf Funding:}}
\providecommand{\funding}[1]{\def\and{{\textperiodcentered} }
\par\addvspace\baselineskip
\noindent\fundingname\enspace\ignorespaces#1}

\begin{document}

\title{Capacitated Network Bargaining Games: Stability and Structure}
\author[1]{Laura Sanità}
\author[2]{Lucy Verberk\(^{(\text{\Letter})}\)}
\affil[1]{
    Bocconi University of Milan, Italy. \protect \\
    {\normalsize\tt{laura.sanita@unibocconi.it}}
}
\affil[2]{
    Eindhoven University of Technology, Netherlands. \protect \\
    {\normalsize\tt{l.p.a.verberk@tue.nl}}
}
\date{}
\maketitle

\begin{abstract}
    Capacitated network bargaining games are popular combinatorial games that involve the structure of matchings in graphs. We show that it is always possible to stabilize unit-weight instances of this problem (that is, ensure that they admit a stable outcome) via capacity-reduction and edge-removal operations, without decreasing the total value that the players can get.

    Furthermore, for general weighted instances, we show that computing a minimum amount of vertex-capacity to reduce to make an instance stable is a polynomial-time solvable problem. We then exploit this to give approximation results for the NP-hard problem of stabilizing a graph via edge-removal operations.
    
    Our work extends and generalizes previous results in the literature that dealt with a unit-capacity version of the problem, using several new arguments. In particular, while previous results mainly used combinatorial techniques, we here rely on polyhedral arguments and, more specifically, on the notion of \emph{circuits} of a polytope.

    \funding{The authors are supported by the NWO VIDI grant VI.Vidi.193.087.}
    
    \keywords{Matching \and Network bargaining \and Circuits.}
\end{abstract}

\section{Introduction}

This paper focuses on \emph{stabilization} problems for \emph{capacitated} Network Bargaining Games (NBG).

NBG were introduced by Kleinberg and Tardos~\cite{Kleinberg2008Balanced} as an extension of Nash's 2-player bargaining solution~\cite{Nash1950}. Here we are given a graph \(\graphVE\) with edge weights \(w\in\R_{\geq0}^\edgeSet\), where the vertices represent players and the edges the deals that the players can make. Each player can enter in \emph{at most one} deal with one of their neighbors, and together they agree on how to split the value of the corresponding edge. An outcome is then associated with a matching \(\matching\) of \(\graph\) representing the deals, and an allocation vector \(y\in\R_{\geq0}^\vertexSet\) with \(y_{u}+y_{v}=\weight_{uv}\) if \(\shortEdge{u}{v}\in\matching\), and \(y_v=0\) if \(v\) is not matched. An outcome \((\matching,y)\) is called \emph{stable} if no player has an incentive to break the current agreements to enter in a deal with a different neighbor, which formally translates in the condition  \(y_u \geq \max_{v: uv \in \edgeSet\setminus\matching} \CBra{\weight_{uv} - y_v}\), for each player \(u\).

Bateni et al~\cite{Bateni2010Cooperative} introduced a more realistic \emph{capacitated} setting, where players are allowed to enter in more than one deal. This generalization can be modeled by adding vertex capacities \(\capacity\in\Z_{\geq0}^\vertexSet\) to the input. An outcome to the capacitated NBG is now given by a \(\capacity\)-matching \(\matching\) and a vector \(a\in\R_{\geq0}^{2\edgeSet}\) that satisfies \(a_{uv}+a_{vu}=\weight_{uv}\) if \(\shortEdge{u}{v}\in\matching\), and \(a_{uv}=a_{vu}=0\) otherwise. The concept of stable outcome naturally generalizes (see \cite{Bateni2010Cooperative} for a formal definition).

A key property of (capacitated) NBG is that instances admitting a stable outcome have a very nice \emph {LP characterization}, as shown by~\cite{Kleinberg2008Balanced,Bateni2010Cooperative}. Specifically, given an instance \(\graphwc\), there exists a stable outcome for the corresponding game on \(\graph\) if and only if the value of a maximum-weight \(\capacity\)-matching  \(\nucG\) equals the value of a maximum-weight \emph{fractional} \(\capacity\)-matching \(\nufcG\), defined as
\begin{equation}\label{eq:fracmatching}
    \nufcG := \max \CBra{ \weight^\top x : \sum_{u: uv \in E} x_{uv}\leq\capacity_v \; \forall v\in\vertexSet, \; 0\leq x\leq1 }.
\end{equation}
In other words, instances admitting stable outcomes are the ones for which the standard LP relaxation of the maximum-weight \(\capacity\)-matching problem has an optimal integral solution. A graph \(\graph\) for which \(\nucG=\nufcG\) is called \emph{stable}.

It can be easily seen from this characterization that there are instances which do not admit stable solutions, for example odd cycles. For this reason, 
several researchers in the past years investigated so-called \emph{stabilization} problems, where the goal is to turn a given graph into a stable one, via some graph operations (see \cite{Biro10,Konemann12,Bock2015Finding,CHANDRASEKARAN201956,Ahmadian2018Stabilizing,Chandrasekaran2017,Ito2017Efficient,Koh2020Stabilizing,gerstbrein2022stabilization}). 
Two common ways to stabilize graphs are edge- and vertex-removal operations, which have a natural NBG interpretation: they  ensure a stable outcome by blocking interactions and blocking players, respectively. 

In the unit-capacity setting, stabilization via edge- and vertex-removal operations for NBG is quite well understood. In particular, (i) stabilizing a graph by removing a minimum number of edges (called the \emph{edge-stabilizer problem}) is NP-hard, and even hard-to-approximate with a constant factor, but admits an \(O(\Delta)\)-approximation, with \(\Delta\) being the maximum degree of a vertex in the graph~\cite{Gottschalk2018Personal,Bock2015Finding,Koh2020Stabilizing}. Differently, (ii) stabilizing a graph by removing a minimum number of vertices (called the \emph{vertex-stabilizer problem}) is solvable in polynomial-time~\cite{Ahmadian2018Stabilizing,Ito2017Efficient,Koh2020Stabilizing}. Moreover, (iii) both problems for unit-weight graphs exhibit a very nice structural property: optimal solution do not decrease the cardinality of a maximum matching, meaning that there is always a way to stabilize the graph without decreasing the total value that the players can get~\cite{Bock2015Finding,Ahmadian2018Stabilizing}. 
The study of stabilization problems in the capacitated setting was recently initiated in~\cite{gerstbrein2022stabilization}, where (among other things) it is shown that the vertex-stabilizer problem becomes NP-hard. 
However, to complete the picture regarding stabilization problems in capacitated NBG, a few questions remain:
\emph{
\begin{enumerate}
    \item[(i)] Can one efficiently stabilize graphs by reducing the capacity of vertices, instead of removing them?
    \item[(ii)] Are there non-trivial approximation algorithms for the edge-stabilizer problem in capacitated NBG instances?
    \item[(iii)] For capacitated NBG with unit-weights, can one still hope to stabilize a graph without decreasing the total value that the players can get?
\end{enumerate}
}

\paragraph*{Our results and techniques.} In this paper, we give an affirmative answer to the above three questions. 

Our work started by realizing that the hardness proved in~\cite{gerstbrein2022stabilization} for the vertex-stabilizer problem crucially relies on the fact that vertices have different capacity values and get removed \emph{completely}. Another (still natural) way to generalize the vertex-stabilizer problem from the unit-capacity setting, is to \emph{decrease} the capacity of the vertices (that is, reduce the number of potential deals that a player can be engaged in). We prove that computing a minimum amount of vertex capacity to reduce to make an instance stable (which we call the \emph{capacity-stabilizer problem}) is a polynomial-time solvable problem (results in \Cref{sec: capacity-stabilizer}).
Interestingly, our solution preserves the total value that the players can get  up to a factor that is asymptotically best possible, and it reduces the capacity of each vertex by at most one. This has a nice network bargaining interpretation: there is always an optimal and at the same time \emph{fair} way to stabilize instances, as no player will have its capacity dramatically reduced compared to others. 

In addition to answering question (i) in a positive way, the algorithm to solve the capacity-stabilizer problem becomes instrumental when dealing with edge-stabilizers. In fact, we crucially exploit it to extend the \(O(\Delta)\)-approximation for NBG to the capacitated settings, answering question (ii)
(results in \Cref{sec: edge-stabilizer}).

Eventually, we manage to show that the key structural property of minimum stabilizers for instances with unit-weights still holds. Namely, that optimal solutions to both the capacity-stabilizer and the edge-stabilizer problem do not decrease the total value that the players can get (results in \Cref{sec: capacity-stabilizer,sec: edge-stabilizer}).

Besides extending the previous known results about NBG to the capacitated setting, what we find interesting are the new arguments we rely on in our proofs. Previous results mainly used combinatorial techniques for both the structural result in (iii) and the main algorithmic ingredient behind (i) and (ii) (which is the fact that the minimum number of fractional odd cycles in the support of an optimal fractional solution to \eqref{eq:fracmatching} provides a \emph{lower bound} on the size of a stabilizer). To extend such results in the presence of capacities, one might try to reduce a capacitated instance to a unit-capacity instance by introducing copies of vertices that have a capacity higher than \(1\), and then use the already existing stabilization results for unit-capacity instances. We show in \Cref{appx: reduction} that such reductions cannot be applied in a straightforward manner. We here instead rely on (new) polyhedral arguments and, in particular, on the notion of \emph{circuits} of a polytope, which are a key concept in optimization (see \Cref{sec: prelim,sec: polyhedral tools}).
Interestingly, our polyhedral view point allows us not only to deal more broadly with capacitated instances, but also to simplify some of the cardinal arguments previously used in the literature: in particular, our lower bound proof is (more general and) much simpler than the corresponding one in~\cite{Koh2020Stabilizing} for the unit-capacity setting.

\section{Preliminaries}\label{sec: prelim}

\paragraph*{Problem definition.}
    Let \(S\) be a multi-set of vertices \(V\). We denote by \(\graph[\capacity_S -1]\) the graph \(\graph\) with the capacity of all vertices in \(S\) reduced by one. Note that if a vertex appears e.g.\ twice in \(S\), its capacity is reduced by two. We define a \emph{capacity-stabilizer} as a multi-set \(S\) of vertices, such that \(\graph[\capacity_S -1]\) is stable. Since \(S\) is a multi-set, the amount of capacity reduced equals the size of \(S\). We define an \emph{edge-stabilizer} as a set \(F \subseteq \edgeSet\), such that \(\graph \setminus F := (\vertexSet,\edgeSet \setminus F)\) is stable.

    \smallskip
    \noindent
    \textbf{Capacity-stabilizer problem:} given \(\graphVE\) with edge weights \(\weight\in\R_{\geq0}^\edgeSet\) and vertex capacities \(\capacity \in \Z^\vertexSet_{\geq 0}\), find a capacity-stabilizer of minimum cardinality.

    \smallskip
    \noindent
    \textbf{Edge-stabilizer problem:} given \(\graphVE\) with edge weights \(\weight\in\R_{\geq0}^\edgeSet\) and vertex capacities \(\capacity \in \Z^\vertexSet_{\geq 0}\), find an edge-stabilizer of minimum cardinality.

\paragraph*{Notation.}
    We use \(\graphwc\) to refer to a graph with edge weights and vertex capacities, and \(\gwcm\) to refer to a graph with a given \(\capacity\)-matching \(\matching\). We say that \(\gwcm\) is stable if \(\graph\) is stable and \(\matching\) is a maximum-weight \(\capacity\)-matching in \(\graph\).
    Let \(S\) be a multi-set of vertices. We use \(\capacity^{S-1}\) to refer to the capacities in \(\graph[\capacity_S - 1]\). For any \(s \in S\), with \(S \setminus s\) we mean removing \(s\) just once from \(S\).
    For a vertex \(v\), we let \(\delta(v)\) be the set of edges of \(\graph\) incident to it. 
    For \(F\subseteq\edgeSet\), we denote by \(d_v^F\) the degree of \(v\) in \(\graph\) with respect to the edges in \(F\). 
    For \(f \in \R^\edgeSet\), we define \(f(F) := \sum_{e\in F} f_e\).
    Given a \(\capacity\)-matching \(\matching\), we say that \(v\in\vertexSet\) is \emph{exposed} if \(d_v^\matching=0\), \emph{unsaturated} if \(d_v^\matching<\capacity_v\) and \emph{saturated} if \(d_v^\matching=\capacity_v\).
    We also use these terms for fractional \(\capacity\)-matchings \(x\), e.g., \(v\) is saturated if \(x(\delta(v)) = \capacity_v\).
    We let \(\symdif\) denote the symmetric difference operator.
    We denote a (\(uv\)-)walk \(W\) by listing its edges and endpoints sequentially, i.e., by \(W=(u; e_1,\ldots,e_k;v)\). We say a walk is closed if \(u=v\). A trail is a walk in which edges do not repeat. A path is a trail in which internal vertices do not repeat. A cycle is a path which starts and ends at the same vertex. 
    Note that the edge set of a walk can be a multi-set.
    For a walk \(W\) (possibly a multi-set) and a \(\capacity\)-matching \(\matching\) (not a multi-set), we define \(W \setminus \matching := \{e \in W : e \notin \matching\}\) and \(W \cap \matching := \{e \in W : e \in \matching\}\). For example, let \(W = \{e_1,e_2,e_2\}\) and \(\matching = \{e_2\}\), then \(W \setminus \matching = \{e_1\}\) and \(W \cap \matching = \{e_2,e_2\}\).

\paragraph*{Duality.}
    The dual of \eqref{eq:fracmatching} is given by
    \begin{equation} \label{eq:fractional_vertex_cover}
        \taufcG := \min \CBra{ \capacity^\top y + 1^\top z : y_u + y_v + z_{uv} \geq \weight_{uv} \ \forall\shortEdge{u}{v}\in\edgeSet, y\geq0, z\geq0 }.
    \end{equation}
    We refer to feasible solutions \((y,z)\) of \eqref{eq:fractional_vertex_cover} as \emph{fractional vertex covers}. By LP theory, we have \(\nucG\leq\nufcG=\taufcG\). The complementary slackness conditions of \(\nufcG\) and \(\taufcG\) are
    \begin{equation}
        \label{eq: compl slack}
        (x_{uv} = 0 \lor y_u + y_v + z_{uv} = \weight_{uv}) 
        \land 
        (y_v = 0 \lor x(\delta(v)) = \capacity_v) 
        \land 
        (z_{uv} = 0 \lor x_{uv} = 1).
    \end{equation}

\subsection{Augmenting Structures}
    \begin{definition}
        We say that a walk \(W\) is \(\matching\)-alternating (w.r.t. a \(\capacity\)-matching \(\matching\)) if its edges are alternating between \(\matching\) and \(\edgeSet \setminus \matching\). We say \(W\) is \(\matching\)-augmenting if it is \(\matching\)-alternating and \(\weight(W\setminus\matching)>\weight(W\cap\matching)\). An \(\matching\)-alternating \(u v\)-walk \(W\) is \emph{proper} if \(W\symdif\matching\) is a \(\capacity\)-matching.
    \end{definition}

    \begin{definition}
        Given an \(\matching\)-alternating walk \(W=(u;e_1,\ldots,\allowbreak e_k;v)\) and an \(\varepsilon>0\), the \emph{\(\varepsilon\)-augmentation} of \(W\) is the vector \(x^{\matching/W}(\varepsilon)\in\R^\edgeSet\) given by
        \begin{equation*}
            x_e^{\matching/W}(\varepsilon) = \begin{cases} 
                1-\kappa(e)\varepsilon & \text{ if } e\in\matching, \\
                \kappa(e)\varepsilon & \text{ if } e\notin\matching,
            \end{cases}
        \end{equation*}
        where \(\kappa(e)=\Abs{\CBra{i\in[k] : e_i=e}}\). We say that \(W\) is \emph{feasible} if there exists an \(\varepsilon>0\) such that the corresponding \(\varepsilon\)-augmentation of \(W\) is a fractional \(\capacity\)-matching.
    \end{definition}

    To get a better understanding of proper and feasible, we state for different kinds of walks what proper and feasible actually come down to.
    \emph{(i)~Non-closed walks:} an \(\matching\)-alternating walk \(W=(u; e_1,\ldots,e_k;v)\), where \(u\neq v\), is proper and feasible if either \(e_1\in\matching\) or \(d_u^\matching\leq\capacity_u-1\), and if either \(e_k\in\matching\) or \(d_v^\matching\leq\capacity_v-1\).
    \emph{(ii)~Even-length closed walks:} an \(\matching\)-alternating walk \(W=(v; e_1,\ldots,e_k;v)\), such that \(k\) is even, is always proper and feasible.
    \emph{(iii)~Odd-length closed walks:} an \(\matching\)-alternating walk \(W=(v; e_1,\ldots,e_k;v)\), such that \(k\) is odd, is proper if either \(e_1,e_k\in\matching\) or \(d_v^\matching\leq\capacity_v-2\), and feasible if either \(e_1,e_k\in\matching\) or \(d_v^\matching\leq\capacity_v-1\).
    
    \begin{theorem}
        [theorem 1 in \cite{gerstbrein2022stabilization}]
        \label{thm: M max iff no augmenting trail}
        A \(\capacity\)-matching \(\matching\) in \(\graphwc\) is maximum-weight if and only if \(\graph\) does not contain a proper \(\matching\)-augmenting trail.
    \end{theorem}

    \begin{theorem}
        [theorem 3 in \cite{gerstbrein2022stabilization}]
        \label{thm: feasible augmenting walk not stable}
        \(\gwcm\) is stable if and only if \(\graph\) does not contain a feasible \(\matching\)-augmenting walk.
    \end{theorem}

\subsection{Basic Fractional c-Matchings}
    The polytope of fractional \(\capacity\)-matchings in \(\graph\) is \(\Pfcm(\graph)\), formally defined as
    \begin{equation*}
        \label{eq: FCM polytope}
        \Pfcm(\graph) := \CBra{x\in\R^\edgeSet : x(\delta(v))\leq\capacity_v \ \forall v\in\vertexSet, 0\leq x\leq1}.
    \end{equation*}
    We write \(\Pfcm\) if \(G\) is clear from the context, or irrelevant.
    We refer to the vertices of \(\Pfcm\) as \define{basic} fractional \(\capacity\)-matchings.
    The next result is well known, see e.g.\ Theorem 21 of \cite{Appa2006bidirected} for half-integrality, but we provide a proof in \Cref{appx: basic f-c-matching} for completeness.
    
    \begin{restatable}{theorem}{basicfraccmatching}
        \label{thm: basic fractional c-matching}
        A fractional \(\capacity\)-matching \(x\) is basic if and only if its components are equal to \(0\), \(\frac{1}{2}\) or \(1\), and the edges with \(x_e=\frac{1}{2}\) induce vertex-disjoint odd cycles with saturated vertices.
    \end{restatable}

    We partition the support of a basic fractional \(\capacity\)-matching \(x\) into the odd cycles induced by \(x_e=\frac{1}{2}\)-edges: \(\oddcyclesx{x}=\CBra{C_1,\ldots,C_q}\) (later referred to as fractional odd cycles), and matched edges: \(\matchedx{x}=\CBra{e\in\edgeSet: x_e=1}\).

    \begin{definition}
        \emph{Alternate rounding \(C = \Par{v;e_1, \ldots, e_{2k+1};v} \in \oddcyclesx{x}\) exposing \(v\)} means replacing \(x_e\) by \(\hat{x}_e=0\) for all \(e\in\CBra{e_1,e_3,\ldots,e_{2k+1}}\) and by \(\hat{x}_e=1\) for all \(e\in\CBra{e_2,e_4,\ldots,e_{2k}}\). Similarly, we define \define{alternate rounding \(C\in\oddcyclesx{x}\) covering \(v\)}.
    \end{definition}
    
    Let \(\cX\) be the set of basic maximum-weight fractional \(\capacity\)-matchings in \(\graph\). Define \(\gamma(\graph):=\min_{x\in\cX} \Abs{\oddcyclesx{x}}\), as the minimum number of fractional odd cycles in the support of any basic maximum-weight fractional \(\capacity\)-matching in \(\graph\). As already noted by \cite{Koh2020Stabilizing} for the unit-capacity case, we have the following.
    \begin{proposition}
        \label{prop: stable iff gammaG=0}
        A graph \(\graphwc\) is stable if and only if \(\gamma(\graph)=0\).
    \end{proposition}
    
    \cite{Koh2020Stabilizing} proposes an algorithm to obtain a maximum-weight basic fractional matching with minimum number of fractional odd cycles (\(\Abs{\oddcyclesx{x}} = \gamma(\graph)\)). Their result can be generalized to \(\capacity\)-matchings. 
    We just state the result here, all details can be found in \Cref{appx: min nr odd cycles}.
    \begin{restatable}{theorem}{minnroddcycles}
        \label{thm: min nr odd cycles}
        A basic maximum-weight fractional \(\capacity\)-matching \(x\) with \(\Abs{\oddcyclesx{x}} = \gamma(\graph)\) can be computed in polynomial time.
    \end{restatable}

\subsection{Circuits of the Fractional c-Matching Polytope}
    Let \(e\) be an edge of a polyhedron \(\cP = \CBra{x\in\R^n: Ax=b, Bx\leq d}\), where \(A\) and \(B\) are integral matrices, and \(b\) and \(d\) are rational vectors. The \define{edge direction} of \(e\) is \(v-w\) for any two distinct points \(v\) and \(w\) on \(e\). The \define{circuits} of a polyhedron (described by \(A\) and \(B\)) are all potential edge directions that can appear for any choice of rational \(b\) and \(d\) \cite[theorem 1.8]{Finhold2014Circuit}. Let \(\cC(\cP)\) denote the set of circuits of \(\cP\) with co-prime integer components.
    
    For a characterization of the circuits of the  fractional \(\capacity\)-matching polytope we rely on \cite{Loera2019Pivot}, which defined five classes of graphs \((\cE_1,\cE_2,\cE_3,\cE_4,\cE_5)\), listed below.
    \begin{enumerate}[(i)]
        \item Let \(\cE_1\) denote the set of all subgraphs \(F\subseteq\graph\) such that \(F\) is an even cycle.
        \item Let \(\cE_2\) denote the set of all subgraphs \(F\subseteq\graph\) such that \(F\) is an odd cycle.
        \item Let \(\cE_3\) denote the set of all subgraphs \(F\subseteq\graph\) such that \(F\) is a path.
        \item Let \(\cE_4\) denote the set of all subgraphs \(F\subseteq\graph\) such that \(F=C\cup P\), where \(C\) is an odd cycle, and \(P\) is a non-empty path that intersects \(C\) only at one endpoint.
        \item Let \(\cE_5\) denote the set of all subgraphs \(F\subseteq\graph\) such that \(F=C_1\cup P\cup C_2\), where \(C_1\) and \(C_2\) are odd cycles, and \(P\) is a path satisfying the following: if \(P\) is non-empty, then \(C_1\) and \(C_2\) are vertex-disjoint and \(P\) intersects each \(C_i\) exactly at its endpoints; if \(P\) is empty then \(C_1\) and \(C_2\) intersect at only one vertex.
    \end{enumerate}
    A set of circuits can be associated to the subgraphs in these classes, by defining:
    \begin{equation*}
    \begin{array}{llclc}
        \cC_1 = \bigcup_{F\in\cE_1}\left\{g\in\CBra{-1,0,1}^\edgeSet: \right.
            & g(e)\neq0         & \quad & \text{iff } e\in\edgeSet(F) \\
            & g(\delta(v))=0    & \quad & \forall v\in\vertexSet(F) 
        & \left.\vphantom{\CBra{1}^E}\right\},\\
        \cC_2 = \bigcup_{F\in\cE_2}\left\{g\in\CBra{-1,0,1}^\edgeSet: \right.
            & g(e)\neq0         & \quad & \text{iff } e\in\edgeSet(F) \\
            & g(\delta(w))\neq0 & \quad & \text{for one } w\in\vertexSet(F) \\
            & g(\delta(v))=0    & \quad & \forall v\in\vertexSet(F)\setminus\CBra{w}
        & \left.\vphantom{\CBra{1}^E}\right\},\\
        \cC_3 = \bigcup_{F\in\cE_3}\left\{g\in\CBra{-1,0,1}^\edgeSet: \right.
            & g(e)\neq0         & \quad & \text{iff } e\in\edgeSet(F) \\
            & g(\delta(v))=0    & \quad & \forall v: \Abs{\delta(v) \cap \edgeSet(F)}=2
        & \left.\vphantom{\CBra{1}^E}\right\},\\
        \cC_4 = \bigcup_{F=(C\cup P)\in\cE_4}\left\{g\in\Z^\edgeSet: \right.
            & g(e)\neq0         & \quad & \text{iff } e\in\edgeSet(F) \\
            & g(\delta(v))=0    & \quad & \forall v: \Abs{\delta(v) \cap \edgeSet(F)}\geq2 \\
            & g(e)\in\CBra{-1,1}& \quad & \forall e\in\edgeSet(C) \\
            & g(e)\in\CBra{-2,2}& \quad & \forall e\in\edgeSet(P)
        & \left.\vphantom{\Z^E}\right\},\\
        \cC_5 = \bigcup_{F=(C_1\cup P\cup C_2)\in\cE_5}\left\{g\in\Z^\edgeSet: \right.
            & g(e)\neq0         & \quad & \text{iff } e\in\edgeSet(F) \\
            & g(\delta(v))=0    & \quad & \forall v\in\vertexSet(F) \\
            & g(e)\in\CBra{-1,1}& \quad & \forall e\in\edgeSet(C_1\cup C_2) \\
            & g(e)\in\CBra{-2,2}& \quad & \forall e\in\edgeSet(P)
        & \left.\vphantom{\Z^E}\right\}.
    \end{array}
    \end{equation*}
    The authors of \cite{Loera2019Pivot} showed that \(\cC_1\cup\cC_2\cup\cC_3\cup\cC_4\cup\cC_5\) is the set of circuits of the \emph{fractional matching polytope}, that is \(\Pfcm\) with \(\capacity = 1\) and without the (redundant) constraints \(x \leq 1\).
    Since the set of circuits stays the same if you change the right hand side vector
    (note that the constraints \(x\leq 1\) are parallel to \(x\geq0\)), the same set of circuits apply to  \(\cC(\Pfcm)\). Hence
    \begin{proposition}\label{cor: circuits of frac c-matching polytope}
        \(\cC(\Pfcm) = \cC_1\cup\cC_2\cup\cC_3\cup\cC_4\cup\cC_5\).
    \end{proposition}
\section{Key Polyhedral Tools}\label{sec: polyhedral tools}

\begin{theorem}
    \label{thm: one step over polyhedron general}
    Let \(\polyhedron\) be any polytope, \(a^\top x \leq b\) be an inequality of the description of \(\polyhedron\), and \(\delta \in \R_{>0}\). Let \(\overline{x}\) be an optimal solution of the LP \(\max \{c^\top x : x \in \polyhedron, a^\top x \leq b - \delta\}\), and assume that \(\overline{x}\) is a non-optimal vertex of the LP \(\max \CBra{c^\top x : x \in \polyhedron}\). Furthermore, assume that there is no vertex \(\widetilde{x}\) of \(\polyhedron\) satisfying \(b-\delta < a^\top \widetilde{x} < b\).
    Then it is possible to move to an optimal solution of \(\max\CBra{c^\top x : x \in \polyhedron}\) from \(\overline{x}\) in one step over the edges of \(\polyhedron\)
    (i.e., there is an optimal vertex of \(\polyhedron\) adjacent to \(\overline{x}\)).
\end{theorem}
\begin{proof}
    Let \(x^*\) be the optimal solution of \(\max\CBra{c^\top x : x \in \polyhedron}\) that is the \emph{closest} vertex to \(\overline{x}\) on \(\polyhedron\) (that is, such that we need a minimum number of steps over the edges of \(\polyhedron\) to reach \(x^*\) from \(\overline{x}\)). Note that \(a^\top \overline{x} = b - \delta\) and \(a^\top x^* = b\), otherwise \(\overline{x} + \lambda(x^* - \overline{x})\), for some small \(\lambda > 0\), and \(x^*\), respectively, contradict the optimality of \(\overline{x}\).
    We need to show that \(\overline{x}\) and \(x^*\) are adjacent on \(\polyhedron\).
    
    Let \(\polyhedron' = \CBra{x \in \polyhedron: a^\top x \geq b - \delta}\). Then \(\overline{x}, x^* \in \polyhedron'\). Note that \(\overline{x}\) and \(x^*\) are adjacent on \(\polyhedron\) if and only if they are adjacent on \(\polyhedron'\). So for the remainder of the proof we restrict ourselves to \(\polyhedron'\).
    
    For the sake of contradiction assume that \(\overline{x}\) and \(x^*\) are not adjacent on \(\polyhedron'\). Then, the line segment of all their convex combinations: \(\lambda \overline{x} + (1-\lambda)x^*\) for \(0 \leq \lambda \leq 1\), is not an edge of \(\polyhedron'\). Hence, any point \(\lambda' \overline{x} + (1 - \lambda') x^*\) for a fixed \(0 < \lambda' < 1\) is also a convex combination of other vertices of \(\polyhedron'\):
    \(
        \lambda' \overline{x} + (1 - \lambda') x^* = \sum_i \alpha_i \hat{x}_i + \sum_j \beta_j \widetilde{x}_j,
    \)
    where \(\alpha_i\geq0\) for all \(i\), \(\beta_j\geq0\) for all \(j\), \(\sum_i \alpha_i + \sum_j \beta_j = 1\), \(\hat{x}_i\) is a vertex of \(\polyhedron'\) with \(a^\top \hat{x}_i = b - \delta\) for all \(i\), and \(\widetilde{x}_j\) is a vertex of \(\polyhedron'\) with \(a^\top \widetilde{x}_j = b\) for all \(j\). If we multiply both sides by \(a\) we get
    \begin{align*}
        && a^\top \Par{\lambda' \overline{x} + (1 - \lambda') x^*} &= a^\top \Par{\textstyle \sum_i \alpha_i \hat{x}_i + \sum_j \beta_j \widetilde{x}_j}, \\
        &\iff& \lambda' (b - \delta) + (1 - \lambda') b &= \textstyle \sum_i \alpha_i (b - \delta) + \sum_j \beta_j b,  \\
        &\iff& b - \lambda' \delta &= \textstyle \Par{\sum_i \alpha_i + \sum_j \beta_j} b - \sum_i \alpha_i \delta,
    \end{align*}
    hence \(\lambda' = \sum_i \alpha_i\), and consequently \(1 - \lambda' = \sum_j \beta_j\). We can also multiply both sides by \(c\). Here we use that \(\overline{x}\) is an optimal solution of \(\max \{c^\top x : x \in \polyhedron, a^\top x \leq b - \delta\}\), and that \(x^*\) is an optimal solution of \(\max \CBra{c^\top x : x \in \polyhedron}\).
    \begin{align*}
        c^\top \Par{\lambda' \overline{x} + (1 - \lambda') x^*} 
            &= \textstyle c^\top \Par{\sum_i \alpha_i \hat{x}_i + \sum_j \beta_j \widetilde{x}_j} 
            = \sum_i \alpha_i c^\top \hat{x}_i + \sum_j \beta_j c^\top \widetilde{x}_j \\
            &\leq \textstyle \sum_i \alpha_i c^\top \overline{x} + \sum_j \beta_j c^\top x^* 
            = \lambda' c^\top \overline{x} + (1 - \lambda') c^\top x^*
    \end{align*}
    So we must have equality throughout. In particular, \(c^\top \widetilde{x}_j = c^\top x^*\), i.e., all \(\widetilde{x}_j\) are optimal solutions to \(\max\CBra{c^\top x : x \in \polyhedron}\). Note that all \(\widetilde{x}_j\)'s are also vertices of \(\polyhedron\). We show that we can choose some \(\widetilde{x}_j\) to be adjacent to \(\overline{x}\) on \(\polyhedron'\), and hence also on \(\polyhedron\), contradicting that \(x^*\) is the optimal solution closest to \(x\).

    Let \(x'\) be a vertex of \(\polyhedron'\) that is adjacent to \(\overline{x}\), such that \(a x' = b\) (such an \(x'\) must exist). Consider the line segment between \(x'\) and \(\lambda' \overline{x} + (1 - \lambda') x^*\): \(\mu x' + (1 - \mu) \Par{\lambda' \overline{x} + (1 - \lambda') x^*}\) for \(0 \leq \mu \leq 1\). For \(\mu < 0\), this line segment extends beyond \(\lambda' \overline{x} + (1 - \lambda') x^*\). If this is still in \(\polyhedron'\), we can write \(\lambda' \overline{x} + (1 - \lambda') x^*\) as a convex combination of \(x'\) and some other \(\hat{x}_i\)'s and \(\widetilde{x}_j\)'s.
    Since \(a x' = b\), by our previous discussion we find that \(x'\) is optimal, reaching our desired contradiction.
    Otherwise, \(\lambda' \overline{x} + (1 - \lambda') x^*\) must be at the boundary, a face, of \(\polyhedron'\). Because \(\lambda' \overline{x} + (1 - \lambda') x^*\) is in this face, the whole line segment \(\lambda \overline{x} + (1 - \lambda) x^*\) for \(0 \leq \lambda \leq 1\) must be in this face. We can then repeat the argument, replacing \(\polyhedron'\) by this face. Since this face has strictly smaller dimension than \(\polyhedron'\), we either find a contradiction in one of the iterations, or we reach a face of dimension one, i.e., an edge of \(\polyhedron'\). Since this edge contains the whole line segment \(\lambda \overline{x} + (1 - \lambda) x^*\) for \(0 \leq \lambda \leq 1\), the line segment is the edge, a contradiction.
\end{proof}

We make use of this theorem for \(\Pfcm\) in two settings: to analyze what happens when we reduce the capacity of a vertex, and when we remove an edge, in capacitated NBG instances. For the first setting, we have the following.

\begin{theorem}
    \label{thm: one step over polyhedron vertex variant}
    Let \(\overline{x}\) be a maximum-weight fractional \(\capacity\)-matching in \(\graph[\capacity_v-1]\) for some \(v \in \vertexSet\). If \(\overline{x}\) is basic in \(\graph\), but not maximum-weight in \(\graph\), then it is possible to move to a basic maximum-weight fractional \(\capacity\)-matching in \(\graph\) in one step over the edges of \(\Pfcm(\graph)\).
\end{theorem}

\begin{proof}
    It readily follows from \Cref{thm: one step over polyhedron general} by letting \(\polyhedron\) = \(\Pfcm(\graph)\), \(a^\top x \leq b\) be \(x(\delta(v)) \leq \capacity_v\), \(\delta = 1\), and \(\weight\) be the objective function.
\end{proof}

In the second setting, we need to do a bit of extra work.

\begin{theorem}
    \label{thm: one or two steps over polyhedron edge variant}
    Let \(\overline{x}\) be a maximum-weight fractional \(\capacity\)-matching in \(\graph \setminus e\) for some \(e \in \edgeSet\). If \(\overline{x}\) is basic in \(\graph\), but not maximum-weight in \(\graph\), then it is possible to move to a basic maximum-weight fractional \(\capacity\)-matching in \(\graph\) in at most two steps over the edges of \(\Pfcm(\graph)\). If two steps are needed, the first one moves to a vertex with \(x_e = \frac12\), and the second one moves to a vertex with \(x_e = 1\).
\end{theorem}
\begin{proof}
    \emph{Case 1: \(x_e \in \CBra{0,1}\) for all vertices of \(\Pfcm(\graph)\).}
    It follows directly from \Cref{thm: one step over polyhedron general} that only one step is needed, by letting \(\polyhedron = \Pfcm(\graph)\), \(a^\top x \leq b\) be \(x_e \leq 1\), \(\delta = 1\), and \(\weight\) the objective function.
    
    \emph{Case 2: there are vertices of \(\Pfcm(\graph)\) that satisfy \(x_e = \frac12\).}
    In this case, let us consider two polytopes:  \(\polyhedron^\leq = \CBra{x \in \Pfcm(\graph) : x_e \leq \frac12}\) and \(\polyhedron^\geq = \CBra{x \in \Pfcm(\graph) : x_e \geq \frac12}\). Let \(\overline{x}\) be a maximum-weight fractional \(\capacity\)-matching in \(\graph \setminus e\), such that \(\overline{x}\) is basic in \(\graph\), but \(\overline{x}\) does not have maximum-weight in \(\graph\). Note that \(\overline{x}\) does not have maximum-weight over \(\polyhedron^\leq\), otherwise it would be of maximum-weight also over \(\Pfcm(\graph)\), contradicting our assumptions. In addition, since \(\overline{x}\) is a vertex of \(\Pfcm(\graph)\), and feasible in \(\polyhedron^\leq\), it is a vertex of \(\polyhedron^\leq\).

    Let \(\polyhedron = \polyhedron^\leq\), \(a x \leq b\) be \(x_e \leq \frac12\), \(\delta = \frac12\), and \(\weight\) be the objective function. We can then apply \Cref{thm: one step over polyhedron general}: it is possible to move to an optimal solution \(\hat{x}\) of \(\max\CBra{ \weight^\top x : x \in \polyhedron^\leq }\) from \(\overline{x}\) in one step over the edges of \(\polyhedron^\leq\). Note that \(\hat{x}_e = \frac12\). 
    If \(\hat{x}\) is a basic maximum-weight fractional \(\capacity\)-matching in \(\graph\), we are done. So suppose that that is not the case.
    
    \emph{Subcase 2a: \(\hat{x}\) is a vertex of \(\Pfcm(\graph)\).} First, note that since \(\hat{x}\) is a vertex of \(\Pfcm(\graph)\), then the edge of \(\polyhedron^\leq\) that is used to move from \(\overline{x}\) to \(\hat{x}\), is also an edge of \(\Pfcm(\graph)\).
    Furthermore, since \(\hat{x}\) is feasible in \(\polyhedron^\geq\), it is also a vertex in \(\polyhedron^\geq\). By assumption, \(\hat{x}\) is optimal over \(\polyhedron^\leq\), so also over \(\polyhedron^\geq\) with the additional constraint \(x_e \leq \frac12\), but not optimal over \(\Pfcm(\graph)\), so also not over \(\polyhedron^\geq\).
    Let \(\polyhedron = \polyhedron^\geq\), \(a x \leq b\) be \(x_e \leq 1\), \(\delta = \frac12\), and \(\weight\) the objective function.
    We can then apply \Cref{thm: one step over polyhedron general}: it is possible to move to an optimal solution \(x^*\) of \(\max\CBra{ \weight^\top x : x \in \polyhedron^\geq }\) from \(\hat{x}\) in one step over the edges of \(\polyhedron^\geq\). Note that \(x^*_e = 1\). Then, \(x^*\) is also an optimal solution of \(\max\CBra{\weight^\top x : x \in \Pfcm(\graph)}\), and a vertex of \(\Pfcm(\graph)\). Since \(\hat{x}\) and \(x^*\) are both vertices of \(\Pfcm(\graph)\), the edge of \(\polyhedron^\geq\) that is used, is also an edge of \(\Pfcm(\graph)\).
    All in all, we get: starting from \(\overline{x}\), it is possible to move to a basic maximum-weight fractional \(\capacity\)-matching in \(\graph\) in two steps over the edges of \(\Pfcm(\graph)\), such that \(x_e = \frac12\) after the first step, and \(x_e = 1\) after the second step.
    
    \emph{Subcase 2b: \(\hat{x}\) is not a vertex of \(\Pfcm(\graph)\).}
    In this case,  we moved from \(\overline{x}\) to \(\hat{x}\) over an edge of \(\polyhedron^\leq\) which is strictly contained in an edge of \(\Pfcm(\graph)\): it must therefore be that \(\polyhedron^\leq\) and \(\polyhedron^\geq\) split this edge in two, and the splitting point, \(\hat{x}\), is a vertex of both polytopes. By assumption, \(\hat{x}\) is optimal over \(\polyhedron^\leq\), so also over \(\polyhedron^\geq\) with the additional constraint \(x_e \leq \frac12\). Since we reached \(\hat{x}\) by moving over just part of an edge of \(\Pfcm(\graph)\), and this increased the weight, moving further along this edge will increase the weight even further. Hence, \(\hat{x}\) is not optimal over \(\polyhedron^\geq\).
    Let \(\polyhedron = \polyhedron^\geq\), \(a x \leq b\) be \(x_e \leq 1\), \(\delta = \frac12\), and \(\weight\) the objective function.
    We can again apply \Cref{thm: one step over polyhedron general}: it is possible to move to an optimal solution \(x^*\) of \(\max\CBra{ \weight^\top x : x \in \polyhedron^\geq }\) from \(\hat{x}\) in one step over the edges of \(\polyhedron^\geq\). Note that \(x^*_e = 1\). Then, \(x^*\) is also an optimal solution of \(\max\CBra{ \weight^\top x : x \in \Pfcm(\graph) }\), and a vertex of \(\Pfcm(\graph)\). 
    Since \(x^*\) is a vertex of \(\Pfcm(\graph)\), but \(\hat{x}\) is not, the edge of \(\polyhedron^\geq\) that is used, is only part of an edge of \(\Pfcm(\graph)\). In particular, it must be the remainder of the edge over which we moved in the first step.
    All in all, we get: starting from \(\overline{x}\), it is possible to move to a basic maximum-weight fractional \(\capacity\)-matching in \(\graph\) in one step over the edges of \(\Pfcm(\graph)\).
\end{proof}

The following theorem is based on methods used in \cite[section III-G]{Sanita2018Diameter}. 
\begin{theorem}
    \label{thm: adjacent vertices polytope fcm}
    If \(x\) and \(y\) are adjacent vertices of \(\Pfcm(\graph)\), then \(y=x+\alpha g\), where \(g \in \cC(\Pfcm)\)  and \(\alpha\in\CBra{\frac{1}{2},1}\). Furthermore, 
    \begin{itemize}
        \item 
            if \(\alpha=1\), then \(g\in\cC_1\cup\cC_2\cup\cC_3\) and \(\Abs{\oddcyclesx{y}}= \Abs{\oddcyclesx{x}}\).
        \item 
            if \(\alpha=\frac{1}{2}\), then \(g\in\cC_1\cup\cC_2\cup\cC_4\cup\cC_5\), and
            \begin{itemize}
                \item 
                    if \(g\in\cC_1\), then \(\Abs{\oddcyclesx{y}}= \Abs{\oddcyclesx{x}}\). 
                \item
                    if \(g\in\cC_2\cup\cC_4\), then \(\Abs{\oddcyclesx{y}}= \Abs{\oddcyclesx{x}} \pm 1\), and the odd cycle in \(g\) belongs to either \(\oddcyclesx{x}\) or \(\oddcyclesx{y}\). 
                \item
                    if \(g\in\cC_5\), then \(\Abs{\oddcyclesx{y}}= \Abs{\oddcyclesx{x}} \pm \CBra{0,2}\), and the odd cycles in \(g\) both belong to \(\oddcyclesx{x}\), or both to \(\oddcyclesx{y}\), or exactly one belongs to \(\oddcyclesx{x}\) and the other to \(\oddcyclesx{y}\).
            \end{itemize}
    \end{itemize}
\end{theorem}
Before we go into the proof, we introduce a fractional perfect \(\capacity\)-matching polytope, which will be helpful.
Consider \(\Pfcm\). Add a non-negative slack-variable for each inequality of the form \(x(\delta(v))\leq c_v\). We get a polytope that naturally corresponds to the set of fractional perfect \(\capacity\)-matchings on a modified graph \(\overline{\graph}=(\vertexSet,\edgeSet\cup L)\), obtained from \(\graph\) by adding a loop edge \(\shortEdge{v}{v}\in L\) for each vertex \(v\in\vertexSet\). We define
\begin{equation*}
    \Pfpcm(\overline{\graph}) := \CBra{x\in\R^{\edgeSet\cup L} : x(\delta(v))=c_v \ \forall v\in\vertexSet, x\geq0, x_e\leq1\ \forall e\in E},
\end{equation*}
as the polytope of fractional perfect \(\capacity\)-matchings in \(\overline{\graph}\).

\begin{proof}\emph{Proof of \Cref{thm: adjacent vertices polytope fcm}. }
    Let \(g\) be the edge direction of the edge between \(x\) and \(y\), scaled in such a way that the components of \(g\) are co-prime. Then, clearly, \(y = x + \alpha g\) for some \(\alpha \neq 0\). Without loss of generality, we can assume that \(\alpha > 0\), since \(-g\) is also an edge direction of the same edge. All edge directions are circuits, hence \(g \in \cC(\Pfcm)\), and in particular, \(g \in \cC_1\cup\cC_2\cup\cC_3\cup\cC_4\cup\cC_5\), by \Cref{cor: circuits of frac c-matching polytope}. That means that the components of \(g\) have a magnitude of at least \(1\). Then it follows from \(0 \leq x \leq 1\), that \(\alpha \leq 1\). Likewise, for the circuits with components that have a magnitude of \(2\), it follows that \(\alpha \leq \frac{1}{2}\). Finally, since \(x\) and \(y\) are vertices, i.e.\ they are basic, their components are equal to \(0\), \(\frac12\), or \(1\), which implies \(\alpha\in\CBra{\frac{1}{2},1}\).
    
    \textit{Case 1: \(\alpha=1\).} 
    As discussed, for circuits with components that have a magnitude of \(2\), \(\alpha\leq\frac{1}{2}\). Hence, in this case, \(g\in\cC_1\cup\cC_2\cup\cC_3\). Furthermore, all components of \(\alpha g\) are integral, which means that fractional edges, and in particular the fractional odd cycles, are not affected. It follows that \(\Abs{\oddcyclesx{y}}= \Abs{\oddcyclesx{x}}\). 
    
    \textit{Case 2: \(\alpha=\frac{1}{2}\).} 
    Circuits in \(\cC_3\) correspond with paths. For either endpoint of this path, applying the circuit results in \(\pm\alpha\cdot1=\pm\frac{1}{2}\) on a single edge incident with the vertex. Since \(x\) and \(y\) are both basic, this is not possible, so \(g\in\cC_1\cup\cC_2\cup\cC_4\cup\cC_5\).
    
    We extend \(x\) and \(y\) to fractional perfect \(\capacity\)-matchings \(\overline{x}\) and \(\overline{y}\) in \(\overline{\graph}\). This extension is uniquely obtained by setting \(\overline{x}_{vv} = \capacity_v-x(\delta(v))\) for each \(\shortEdge{v}{v}\in L\), likewise for \(\overline{y}\). Since \(x\) and \(y\) are adjacent vertices of \(\Pfcm(\graph)\), it follows that \(\overline{x}\) and \(\overline{y}\) are adjacent vertices of \(\Pfpcm(\overline{\graph})\). We extend \(g\) to \(\overline{g}\) such that \(\overline{y} = \overline{x} +\frac{1}{2}\overline{g}\).
    
    Let \(\edgeSet^1 = \CBra{e\in \edgeSet : \overline{x}_e=\overline{y_e}=1}\), and \(\cG\) be the graph induced by the supports of \(\overline{x}\) and \(\overline{y}\), minus the edges in \(\edgeSet^1\). 
    We claim that there is exactly one component of \(\cG\) that contains an edge \(e\) with \(\overline{x}_e\neq\overline{y}_e\). 
    Clearly there is at least one, since \(x\neq y\) and hence \(\overline{x}\neq\overline{y}\). 
    Actually, \(\overline{x} \neq \overline{y}\) only on the support of \(\overline{g}\). Since \(\overline{x}_e \neq \overline{y}_e\) for every edge \(e\) in the support of \(\overline{g}\), we have that at least one of \(\overline{x}_e\) and \(\overline{y}_e\) is not zero, and at least one is not one. Hence, all those edges are in \(\cG\), and in particular they all are in the same component, since \(\overline{g}\) is connected.
    
    Let \(\mathcal{K}\) be the component of \(\cG\) that contains an edge \(e\) with \(\overline{x}_e\neq\overline{y}_e\), and let \(k\) be the number of vertices in this component. Let \(K\) be a subgraph of \(\overline{\graph}\) induced by the vertices in \(\mathcal{K}\), minus the edges in \(\edgeSet^1\), where we change the capacities accordingly: for each vertex \(v \in \vertexSet(K)\), reduce the capacity of \(v\) by \(\Abs{\delta(v) \cap \edgeSet^1}\). Let \(\restr{\overline{x}}{K}\) be obtained from \(\overline{x}\) by restricting to the edges in \(K\), likewise for \(\restr{\overline{y}}{K}\). Note that \(\restr{\overline{x}}{K}\) and \(\restr{\overline{y}}{K}\) are perfect \(\capacity\)-matchings in \(K\) with the adjusted capacity. In particular, they are adjacent vertices of \(\Pfpcm(K)\), since \(\overline{x}\) and \(\overline{y}\) are adjacent vertices of \(\Pfpcm(\overline{\graph})\).
    
    Let \(A\) be the incidence matrix of \(K\). Since the columns associated to the loop edges form an identity matrix, the rank of \(A\) is \(k\). Since \(\restr{\overline{x}}{K}\) and \(\restr{\overline{y}}{K}\) are adjacent vertices, there must be \(\Abs{\edgeSet(K)}-1\) linearly independent constraints that are tight for both of them. Since the rank of \(A\) is \(k\), and we removed the edges for which the ``\(\leq1\)'' constraint is tight for both \(\overline{x}\) and \(\overline{y}\), this implies that there are at least \(\Abs{\edgeSet(K)}-1-k\) edges for which the ``\(\geq 0\)'' constraint  is tight for both of them. Consequently, there are at most \(k+1\) edges in the union of the supports of \(\restr{\overline{x}}{K}\) and \(\restr{\overline{y}}{K}\). Note that the graph induced by the supports of \(\restr{\overline{x}}{K}\) and \(\restr{\overline{y}}{K}\) is \(\mathcal{K}\). With \(k+1\) edges on \(k\) connected vertices, we have a spanning tree, plus two additional (possibly loop) edges: it is easy to realize then that there can be at most two odd cycles in \(\mathcal{K}\).
    
    \textit{Subcase 2a: \(g\in\cC_1\).}
    The support of \(\overline{g}\) is an even cycle, say \(C\). If \(\restr{\overline{x}}{K}=\frac{1}{2}\) for all edges on \(C\), then \(\restr{\overline{x}}{K}\), and therefore also \(x\), contains a fractional even cycle, which contradicts that \(x\) is basic. Similarly, if \(\restr{\overline{x}}{K}\) is integral for all edges on \(C\), \(\restr{\overline{y}}{K}=\frac{1}{2}\) for all edges on \(C\), contradicting that \(y\) is basic. Hence, \(\restr{\overline{x}}{K}\) has edges on \(C\) with integral value, and also edges with value \(\frac{1}{2}\). The fractional edges imply that \(\restr{\overline{x}}{K}\) has at least one fractional odd cycle. The integral edges become fractional in \(\restr{\overline{y}}{K}\), which means \(\restr{\overline{y}}{K}\) also has at least one fractional odd cycle, different from the one of \(\restr{\overline{x}}{K}\). These odd cycles are distinct, and both in \(\mathcal{K}\), and we have already shown that \(\mathcal{K}\) contains at most two odd cycles. Hence, \(\Abs{\oddcyclesx{\restr{\overline{y}}{K}}}= \Abs{\oddcyclesx{\restr{\overline{x}}{K}}}\).
    
    \textit{Subcase 2b: \(g\in\cC_2\cup\cC_4\).}
    The support of \(\overline{g}\) is a (possibly empty) path, an odd cycle and a loop edge, of which only the odd cycle can influence the fractional odd cycles in  \(\restr{\overline{x}}{K}\) and \(\restr{\overline{y}}{K}\). If this odd cycle in the support of \(\overline{g}\) is a fractional odd cycle in \(\restr{\overline{x}}{K}\)/\(\restr{\overline{y}}{K}\), then note that \(\restr{\overline{y}}{K}\) contains exactly one less/more fractional odd cycle than \(\restr{\overline{x}}{K}\). Otherwise, both \(\restr{\overline{x}}{K}\) and \(\restr{\overline{y}}{K}\) have fractional and integral edges on the odd cycle of \(g\). That means that \(\restr{\overline{x}}{K}\) and \(\restr{\overline{y}}{K}\) both have at least one odd cycle in the component, different from the odd cycle in \(g\) and different from each other. But then there are at least three odd cycles in the component, a contradiction. Hence, the odd cycle in \(g\) belongs to either \(\oddcyclesx{\restr{\overline{x}}{K}}\) or \(\oddcyclesx{\restr{\overline{y}}{K}}\), and \(\Abs{\oddcyclesx{\restr{\overline{y}}{K}}}= \Abs{\oddcyclesx{\restr{\overline{x}}{K}}} \pm 1\).
    
    \textit{Subcase 2c: \(g\in\cC_5\).}
    The support of \(\overline{g}\) is two odd cycles connected by a (possibly empty) path. Since \(\mathcal{K}\) contains at most two odd cycles, and \(\overline{g}\) already contains two odd cycles, these are the only odd cycles. Similar to the previous subcase, for both the odd cycles separately we can argue that not both \(\restr{\overline{x}}{K}\) and \(\restr{\overline{y}}{K}\) can have fractional and integral edges on the odd cycle, i.e., each odd cycle belongs to either \(\oddcyclesx{\restr{\overline{x}}{K}}\) or \(\oddcyclesx{\restr{\overline{y}}{K}}\). There are three options: both odd cycles belong to \(\oddcyclesx{\restr{\overline{x}}{K}}\), or both to \(\oddcyclesx{\restr{\overline{y}}{K}}\), or one to \(\oddcyclesx{\restr{\overline{x}}{K}}\) and one to \(\oddcyclesx{\restr{\overline{y}}{K}}\). It follows that \(\Abs{\oddcyclesx{\restr{\overline{y}}{K}}}= \Abs{\oddcyclesx{\restr{\overline{x}}{K}}} \pm \CBra{0,2}\).
     
    Since \(\overline{x}\) equals \(\overline{y}\) outside of \(K\), our conclusions carry over from \(\restr{\overline{x}}{K}\) and \(\restr{\overline{y}}{K}\) to \(\overline{x}\) and \(\overline{y}\). In addition, removing loop edges, i.e., going back from \(\overline{x},\overline{y}\) to \(x,y\), does not affect fractional odd cycles. Hence, our conclusions also hold for \(x\) and \(y\):
    If \(g \in \cC_1\), then \(\Abs{\oddcyclesx{y}}= \Abs{\oddcyclesx{x}}\).
    If \(g \in \cC_2 \cup \cC_4\), then \(\Abs{\oddcyclesx{y}}= \Abs{\oddcyclesx{x}} \pm 1\) and the odd cycle in \(g\) belongs to either \(\oddcyclesx{x}\) or \(\oddcyclesx{y}\). 
    If \(g \in \cC_5\), then \(\Abs{\oddcyclesx{y}}= \Abs{\oddcyclesx{x}} \pm \CBra{0,2}\), and the odd cycles in \(g\) both belong to \(\oddcyclesx{x}\), or both to \(\oddcyclesx{y}\), or exactly one to \(\oddcyclesx{x}\) and the other to \(\oddcyclesx{y}\).
\end{proof}
\section{Capacity-Stabilizer}\label{sec: capacity-stabilizer}
We first exploit the above polyhedral results to prove that a lower bound on the size of a capacity-stabilizer is given by the minimum number of fractional odd cycles in the support of any basic maximum-weight fractional \(\capacity\)-matching.

\begin{lemma}
    \label{lem: lb size of vertex stab}
    For every capacity-stabilizer \(S\), \(\Abs{S} \geq \gamma(\graph)\).
\end{lemma}
\begin{proof}
    To prove the lemma, by \Cref{prop: stable iff gammaG=0}, it is enough to show that reducing the capacity of any vertex by one decreases the number of fractional odd cycles by at most one. Therefore, from now on, we concentrate on proving the following statement: 
    \begin{equation}
        \label{eq: lb gamma(G) reduced capacity}
        \tag{\(\star\)}
        \text{for all } v\in\vertexSet, \ \gamma(\graph[\capacity_v - 1])\geq\gamma(\graph) - 1.
    \end{equation}

    Let \(x\) be a basic maximum-weight fractional \(\capacity\)-matching in \(\graph[\capacity_v-1]\) with \(\gamma(\graph[\capacity_v-1])\) fractional odd cycles. Let \((y,z)\) be an optimal fractional vertex cover in \(\graph[\capacity_v-1]\), satisfying complementary slackness (see \Cref{eq: compl slack}) with \(x\). Note that increasing the capacity does not influence the feasibility of \(x\) and \((y,z)\), hence they are a fractional \(\capacity\)-matching and vertex cover in \(\graph\), respectively. 

    If \(x\) has maximum-weight in \(\graph\), then \(\gamma(\graph) \leq \Abs{\oddcyclesx{x}} = \gamma(\graph[\capacity_v-1])\), and hence  \eqref{eq: lb gamma(G) reduced capacity} holds.

    Assume now that \(x\) does not have maximum-weight in \(\graph\). Then, \(x\) and \((y,z)\) cannot satisfy complementary slackness in \(\graph\). The change from \(\graph[\capacity_v-1]\) to \(\graph\) only influences the complementary slackness condition \(y_v = 0 \lor x(\delta(v))=\capacity_v-1\), so we must have \(y_v>0\) and \(x(\delta(v))=\capacity_v-1<\capacity_v\). We distinguish two cases.
    
    \textit{Case 1: \(v\in\vertexSet(C)\) for some \(C\in\oddcyclesx{x}\).} 
    Create a new fractional \(\capacity\)-matching \(\hat{x}\) by alternate rounding \(C\) covering \(v\). One can check that \(\hat{x}\) and \((y,z)\) satisfy complementary slackness in \(\graph\). Consequently, \(\hat{x}\) is a maximum-weight fractional \(\capacity\)-matching in \(\graph\), and given that \(x\) is basic, so is \(\hat{x}\). 
    Then, \(\gamma(\graph) \leq \Abs{\oddcyclesx{\hat{x}}} = \Abs{\oddcyclesx{x}}-1 = \gamma(\graph[\capacity_v-1]) - 1\), and hence  \eqref{eq: lb gamma(G) reduced capacity} holds.
    
    \textit{Case 2: \(v\notin\vertexSet(\oddcyclesx{x})\).}
    Since \(x\) is basic in \(\graph[\capacity_v-1]\) and \(v\) is not part of any fractional odd cycle, by \Cref{thm: basic fractional c-matching}, \(x\) is also basic in \(\graph\). 
    Then, by \Cref{thm: one step over polyhedron vertex variant}, we can move to a basic maximum-weight fractional \(\capacity\)-matching \(x^*\) in \(\graph\) in one step over the edges of \(\Pfcm(\graph)\). By \Cref{thm: adjacent vertices polytope fcm}, \(x^*=x+\alpha g\) for \(\alpha\in\CBra{\frac{1}{2},1}\) and \(g\in\cC_1\cup\cC_2\cup\cC_3\cup\cC_4\cup\cC_5\). Because \(\weight^\top x^* > \weight^\top x\), \(x^*\) cannot be feasible in \(\graph[\capacity_v-1]\), so we must have \(x^*(\delta(v)) = \capacity_v = x(\delta(v)) + 1\). Consequently, \(g \in \cC_2 \cup \cC_3 \cup \cC_4\). Then, by \Cref{thm: adjacent vertices polytope fcm}, we have \(\Abs{\oddcyclesx{x^*}} = \Abs{\oddcyclesx{x}} \pm \CBra{0,1}\). Therefore \(\Abs{\oddcyclesx{x^*}} \leq \Abs{\oddcyclesx{x}} + 1\), and consequently, \(\gamma(\graph) \leq \Abs{\oddcyclesx{x^*}} \leq \Abs{\oddcyclesx{x}}+1 = \gamma(\graph[\capacity_v-1]) + 1\), yielding \eqref{eq: lb gamma(G) reduced capacity}.
\end{proof}

We can now state the polynomial-time algorithm to solve the stabilization problem via capacity reduction.
\begin{algorithm}[ht]
    \caption{stabilization by capacity reduction}
    \label{alg: capacitated vertex-stab by reducing capacity weighted}
        initialize \(S\leftarrow\emptyset\) \\
        compute a basic maximum-weight fractional \(\capacity\)-matching \(x\) in \(\graph\) with \(\gamma(\graph)\) fractional odd cycles \(C_1,\ldots, C_{\gamma(\graph)}\), and a minimum fractional vertex cover \((y,z)\) in \(\graph\) \\
        \For{\(i=1\) to \(\gamma(\graph)\)}
        {
            \(S\leftarrow S + \arg\min_{v\in\vertexSet(C_i)} y_v\) \\
        }
        \Return \(\graph[\capacity_S-1]\)
\end{algorithm}

\begin{theorem}
    \label{thm: correctnes capacity reduction alg}
     \Cref{alg: capacitated vertex-stab by reducing capacity weighted} is a polynomial-time algorithm that computes a minimum capacity-stabilizer \(S\) for \(\graph\). Moreover: 
    \begin{itemize}
        \item[(a)] The solution \(S\) reduces the capacity of each vertex by at most one unit.  
        \item[(b)] The solution \(S\) preserves the weight of a maximum-weight matching by a factor of \(\frac{2}{3}\), i.e., \(\nu^\capacity(\graph[\capacity_S-1]) \geq \tfrac{2}{3} \nucG\).
    \end{itemize}
\end{theorem}
\begin{proof}
    Let \(S=\CBra{v_1,\ldots,v_{\gamma(\graph)}}\) be the set of vertices whose capacity is reduced in \Cref{alg: capacitated vertex-stab by reducing capacity weighted}. Let \(\hat{x}\) be obtained from \(x\) by alternate rounding \(C_i\) exposing \(v_i\), for all \(i\in\CBra{1,\ldots,\gamma(\graph)}\). 
    Clearly, \(\hat{x}\) is a fractional \(\capacity\)-matching in \(\graph[\capacity_S-1]\). In addition, \((y,z)\) is still a fractional vertex cover in \(\graph[\capacity_S-1]\). One can check that they satisfy complementary slackness with respect to \(\graph[c_S-1]\). Hence, they are optimal in \(\graph[\capacity_S-1]\).
    Note that \(\hat{x}\) is an integral matching. Hence, \(\graph[\capacity_S-1]\) is stable. Moreover, \(\Abs{S}=\gamma(\graph)\), which is minimum by \Cref{lem: lb size of vertex stab}. 

    Since all cycles in \(\oddcyclesx{x}\) are vertex-disjoint, the set \(S\) is \emph{not} a multi-set. Hence, (a) holds.
    To see (b), note that since \(v_i=\arg\min_{v\in\vertexSet(C_i)} y_v\), we have
    \begin{equation*}
        y_{v_i} \leq \frac{y(\vertexSet(C_i))}{\Abs{C_i}} \leq \tfrac{1}{3} y(\vertexSet(C_i)). 
    \end{equation*}
    Then, using stability of \(\graph[\capacity_S-1]\) and optimality of \((y,z)\) in \(\graph[\capacity_S-1]\), we find
    \begin{align*}
        \nu^\capacity(\graph[\capacity_S-1]) 
            &= \tau_f^\capacity(\graph[\capacity_S-1])
            = \Par{\capacity^{S - 1}}^\top y + 1^\top z
            = \capacity^\top y - \sum_{i=1}^{\gamma(\graph)} y_{v_i}  + 1^\top z \\
            &\geq \capacity^\top y - \sum_{i=1}^{\gamma(\graph)} \tfrac{1}{3} y(\vertexSet(C_i)) + 1^\top z
            \geq \tfrac{2}{3} \Par{\capacity^\top y + 1^\top z}
            = \tfrac{2}{3} \taufcG
            \geq \tfrac{2}{3} \nucG.
    \end{align*}
    In the above chain of inequalities, we use the fact that \(\capacity_v-\frac{1}{3}\geq\frac{2}{3}\capacity_v\) for all \(v\in\vertexSet(\oddcyclesx{x})\). This is trivially true  since for all these vertices \(\capacity_v \geq1\).
\end{proof}

The previous theorem shows that there always exists a capacity-stabilizer of minimum size that preserves the total value that the players can get by a factor of \(\frac{2}{3}\). We note that, for arbitrary weighted graphs, this factor is asymptotically best possible, as shown by \cite{Koh2020Stabilizing} already in the unit-capacity case. However, for unit-capacities \emph{and} unit-weights, the authors in~\cite{Ahmadian2018Stabilizing} proved a stronger statement: namely, that inclusion-wise minimal stabilizers completely preserve the total value that the players can get (i.e., up to a factor 1). 
Using our polyhedral tools, we can show that this statement still holds in the capacitated setting (and note that it is satisfied by the solution
provided by our algorithm).

\begin{theorem}
    \label{thm: min vertex-stab size preserving}
    In \(\graphunitc\), for any inclusion-wise minimal capacity-stabilizer \(S\), we have \(\nu^\capacity(\graph[\capacity_S - 1]) = \nucG\).
\end{theorem}
\begin{proof}
    Let \(\matching\) be a maximum-cardinality \(\capacity\)-matching in \(\graph[\capacity_S - 1]\). 

    \begin{claim}
        \(\matching\) is maximum in \(\graph[\capacity_{S \setminus v} - 1]\), for any \(v \in S\).
    \end{claim}
    \begin{proof}
         For the sake of contradiction, suppose that \(\Abs{\matching} < \nu^\capacity(\graph[\capacity_{S \setminus v} - 1])\).

        Since \(S\) is a stabilizer, \(\graph[\capacity_S - 1]\) is stable, and hence there exists a fractional vertex cover \((y,z)\) that satisfies complementary slackness with \(\matching\) in \(\graph[\capacity_S - 1]\). Increasing the capacity of a vertex does not change feasibility of \((y,z)\), hence, \((y,z)\) is a fractional vertex cover in \(\graph[\capacity_{S \setminus v} - 1]\). Observe that, since the edges have unit-weight, we can assume w.l.o.g.\ that \(y \leq 1\). Then
        \begin{equation*}
            \tau_f^\capacity(\graph[\capacity_{S \setminus v} - 1])
                \leq (\capacity^{S - 1})^\top y + y_v + 1^\top z
                = \Abs{\matching} + y_v
                \leq \Abs{\matching} + 1
                \leq \nu^\capacity(\graph[\capacity_{S \setminus v} - 1]),
        \end{equation*}
        i.e., \(\graph[\capacity_{S \setminus v} - 1]\) is stable, contradicting the minimality of \(S\).
    \end{proof}

    \begin{claim}
        \(\matching\) is maximum in \(\graph[\capacity_{S \setminus \CBra{u,v}} - 1]\), for any \(u, v \in S\).
    \end{claim}
    \begin{proof}
        For the sake of contradiction, suppose that \(\Abs{\matching} < \nu^\capacity(\graph[\capacity_{S \setminus \CBra{u,v}} - 1])\).

        Let \(x\) be the indicator vector of \(\matching\), then by \Cref{thm: basic fractional c-matching}, \(x\) is basic.
        Since \(S\) is inclusion-wise minimal, \(\graph[\capacity_{S \setminus v} - 1]\) is not stable, and thus \(x\) is not a maximum fractional \(\capacity\)-matching in \(\graph[\capacity_{S \setminus v} - 1]\). We can apply \Cref{thm: one step over polyhedron vertex variant} to \(x\), \(\graph[\capacity_S - 1]\) and \(\graph[\capacity_{S \setminus v} - 1]\), and conclude that there exists a basic maximum-weight fractional \(\capacity\)-matching \(\hat{x}\) in \(\graph[\capacity_{S \setminus v} - 1]\), which is adjacent to \(x\) on \(\Pfcm\). By \Cref{thm: adjacent vertices polytope fcm}, \(\hat{x} = x + \alpha g\), where \(\alpha\in\CBra{\frac{1}{2},1}\) and \(g\in\cC_1\cup\cC_2\cup\cC_3\cup\cC_4\cup\cC_5\). Since \(\matching\) is maximum in \(\graph[\capacity_{S \setminus v} - 1]\) by the previous claim, \(\hat{x}\) cannot be integral, so \(\alpha = \frac12\), and consequently by \Cref{thm: adjacent vertices polytope fcm}, \(g\notin\cC_3\). Furthermore, the circuits in \(\cC_1 \cup \cC_5\) are not augmenting in cardinality, so \(g\in\cC_2\cup\cC_4\), and necessarily \(v\) must be the only vertex with \(g(\delta(v)) \neq 0\). Observing now that \(\Abs{\matching}\) and \(\nu^\capacity(\cdot)\) are integral, we find that \(\hat{x}\) is not a maximum fractional \(\capacity\)-matching in \(\graph[\capacity_{S \setminus \CBra{u,v}} - 1]\), since
        \begin{equation*}
            1^\top \hat{x}
                 = 1^\top x + \tfrac12
                = \Abs{\matching} + \tfrac12
                < \nu^\capacity(\graph[\capacity_{S \setminus \CBra{u,v}} - 1])
                \leq \nu_f^\capacity(\graph[\capacity_{S \setminus \CBra{u,v}} - 1]),
        \end{equation*}
       Applying \Cref{thm: one step over polyhedron vertex variant} again, we get that there exists a basic maximum-weight fractional \(\capacity\)-matching \(x^*\) in \(\graph[\capacity_{S \setminus \CBra{u,v}} - 1]\), which is adjacent to \(\hat{x}\) on \(\Pfcm\). By \Cref{thm: adjacent vertices polytope fcm}, \(x^* = \hat{x} + \beta h\), where \(\beta\in\CBra{\frac{1}{2},1}\) and \(h\in\cC_1\cup\cC_2\cup\cC_3\cup\cC_4\cup\cC_5\). As before, the circuits in \(\cC_1 \cup \cC_5\) are not augmenting in cardinality, so \(h \in \cC_2\cup\cC_3\cup\cC_4\). Since \(S\) is an inclusion-wise minimal stabilizer, \(\graph[\capacity_{S \setminus \CBra{u,v}} - 1]\) is not stable, consequently \(1^\top x^* > \nu^\capacity(\graph[\capacity_{S \setminus \CBra{u,v}} - 1])\). Using integrality of \(\nu^\capacity(\cdot)\) and half-integrality of \(x^*\) and \(\hat{x}\), this implies \(1^\top x^* \geq 1^\top \hat{x} + 1\). So we must have \(\beta = 1\), hence \(h \in \cC_2 \cup \cC_3\) by \Cref{thm: adjacent vertices polytope fcm}, and \(1^\top x^* = 1^\top \hat{x} + 1\).
        Furthermore, \(u\) must be one of the (possibly two) vertices with \(h(\delta(u))>0\). 
        
        Assume first that \(h \in \cC_2\). Then, \(u\) is the only vertex with \(h(\delta(u))>0\). However, as components of \(\beta h\) have a magnitude of one, it means that
        \(u\) is not saturated in \(\graph[\capacity_{S \setminus v} - 1]\). Therefore, \(\hat{x} + \frac{1}{2} h\) is a fractional \(\capacity\)-matching in 
        \(\graph[\capacity_{S \setminus v} - 1]\) with higher objective value than \(\hat{x}\), contradicting its optimality. 
        
        We are left with \(h \in \cC_3\). It follows that the support of \(h\) is a path \(P_h\) with endpoints \(u\) and some vertex \(t \neq u\). Note that, necessarily, \(t \neq v\), as \(v\) is saturated in \(\hat{x}\) while \(t\) is not. In particular, neither \(u\) nor \(t\) are saturated in \(\graph[\capacity_{S \setminus u} - 1]\). However, we know that \(x + \beta h\) cannot be a fractional \(\capacity\)-matching in \(\graph[\capacity_{S \setminus u} - 1]\), because  \(\matching\) is of maximum cardinality in \(\graph[\capacity_{S \setminus u} - 1]\) by the previous claim. Since \(x + \beta h\) does not violate any capacity bound in \(\graph[\capacity_{S \setminus u} - 1]\), the reason why \(x + \beta h\) is not a fractional \(\capacity\)-matching must be the fact that either \(0 \leq x + \beta h\) or \(x + \beta h \leq 1\) does \emph{not} hold. Since instead  \(0 \leq x + \alpha g + \beta h \leq 1\) holds, it follows that the supports of \(h\) and \(g\) must share some edge. Note that since all components of \(\beta h\) have a magnitude of one, the support of \(h\) cannot overlap with the cycle in the support of \(g\). Let \(\ell\) be the last vertex on the \(ut\)-path \(P_h\) that is an endpoint of a shared edge between the supports of \(h\) and \(g\). By construction, the subpath \(P_1\) from \(\ell\) to \(t\) in \(P_h\) is then an \(\matching\)-alternating path. 
        Let \(P_g\) denote the edges in the support of \(g\), and let \(P_2\) be the path from \(v\) to \(\ell\) in \(P_g\). Note that
        \(P_2\) is also an \(\matching\)-alternating path. Then, one observes that either \(P_2 \cup P_1\) is a proper \(\matching\)-augmenting \(tv\)-path in \(\graph[\capacity_{S \setminus v} - 1]\) (contradicting the previous claim), or \(P_1 \cup (P_g \setminus P_2)\) is a circuit that we can apply to (fractionally) increase the cardinality of \(x\) in \(\graph[\capacity_S - 1]\), contradicting the stability of \(\graph[\capacity_S - 1]\).
    \end{proof}

    Suppose for the sake of contradiction that \(\Abs{\matching} < \nucG\). Then there exists a proper \(\matching\)-augmenting \(st\)-trail \(T\) in \(\graph\), by \Cref{thm: M max iff no augmenting trail} (note that, possibly, \(s=t\)). 
    Since \(\matching\) is maximum in \(\graph[\capacity_S - 1]\), \(T\) cannot be proper in \(\graph[\capacity_S - 1]\), by \Cref{thm: M max iff no augmenting trail}.
    Therefore, \(\Abs{S \cap \CBra{s,t}} \geq 1\). We distinguish two cases.

    \emph{Case 1: \(\Abs{S \cap \CBra{s,t}} = 1\).} Without loss of generality, let \(s\) be the vertex whose capacity gets reduced by \(S\). 
    If \(s \neq t\), then \(\capacity_s^{S \setminus s - 1} = \capacity^{S-1}_s + 1\) and \(\capacity_t^{S \setminus s - 1} = \capacity_t\). If \(s = t\) then \(\capacity_s^{S \setminus s - 1} = \capacity_s\).
    In both cases, \(T\) is a proper \(\matching\)-augmenting trail in \(\graph[\capacity_{S \setminus s} - 1]\), contradicting the first claim.

    \emph{Case 2: \(\Abs{S \cap \CBra{s,t}} = 2\).}
    If \(s \neq t\), then \(\capacity_s^{S \setminus \CBra{s,t} - 1} = \capacity^{S-1}_s + 1\) and \(\capacity_t^{S \setminus \CBra{s,t} - 1} = \capacity^{S-1}_t + 1\). If \(s = t\) then \(\capacity_s^{S \setminus \CBra{s,s} - 1} = \capacity_s^{S-1} + 2\). 
    In both cases, \(T\) is a proper \(\matching\)-augmenting trail in \(\graph[\capacity_{S \setminus \CBra{s,t}} - 1]\), contradicting the second claim.
\end{proof}

\section{Edge-Stabilizer}\label{sec: edge-stabilizer}

In this section we state our results for the edge-stabilizer problem. 
First, we generalize a lower bound on the size of an edge-stabilizer, provided in the unit-capacity setting.

\begin{lemma}
    \label{lem: size of edge-stab}
    For every edge-stabilizer \(F\), \(\Abs{F} \geq \frac{1}{2}\gamma(\graph)\).
\end{lemma}
\begin{proof}
    To prove the lemma, by \Cref{prop: stable iff gammaG=0}, it is enough to show that removing one edge decreases the number of fractional odd cycles by at most two. Therefore, from now on, we concentrate on proving the following statement:
    \begin{equation}
        \label{eq: lb gamma(G) removed edge}
        \tag{\(\diamond\)}
        \text{for all } e \in \edgeSet, \ \gamma(\graph\setminus e) \geq \gamma(\graph) - 2.
    \end{equation}
    
    Let \(x\) be a basic maximum-weight fractional \(\capacity\)-matching in \(\graph \setminus e\) with \(\gamma(\graph \setminus e)\) fractional odd cycles. Extend \(x\) to \(\graph\) by setting \(x_e = 0\). Then \(x\) is a basic fractional \(\capacity\)-matching in \(\graph\). 
    
    If \(x\) has maximum-weight in \(\graph\), then \(\gamma(\graph) \leq \Abs{\oddcyclesx{x}} = \gamma(\graph \setminus e)\), and hence \eqref{eq: lb gamma(G) removed edge} holds. 
    
    If \(x\) does not have maximum-weight in \(\graph\), we can apply \Cref{thm: one or two steps over polyhedron edge variant}: we can move to a basic maximum-weight fractional \(\capacity\)-matching \(x^*\) in \(\graph\) in at most two steps over the edges of \(\Pfcm(\graph)\), and if two steps are needed, the first one moves to a vertex with \(x_e=\frac12\), and the second one to a vertex with \(x_e=1\).  

    Suppose only one step was needed. By \Cref{thm: adjacent vertices polytope fcm}, \(x^*=x+\alpha g\) for \(\alpha\in\CBra{\frac{1}{2},1}\) and \(g\in\cC_1\cup\cC_2\cup\cC_3\cup\cC_4\cup\cC_5\). Then, by \Cref{thm: adjacent vertices polytope fcm}, we have \(\Abs{\oddcyclesx{x^*}} = \Abs{\oddcyclesx{x}} \pm \CBra{0,1,2}\). So definitely, \(\Abs{\oddcyclesx{x^*}} \leq \Abs{\oddcyclesx{x}} + 2\), and consequently, \(\gamma(\graph) \leq \Abs{\oddcyclesx{x^*}} \leq \Abs{\oddcyclesx{x}} + 2 = \gamma(\graph \setminus e) + 2\), yielding \eqref{eq: lb gamma(G) removed edge}.

    Suppose two steps were needed. Let \(\hat{x}\) be the vertex reached after the first step. By \Cref{thm: adjacent vertices polytope fcm}, \(\hat{x} = x+ \alpha g\) and \(x^* = \hat{x} + \beta h\) for \(\alpha,\beta\in\CBra{\frac{1}{2},1}\) and \(g,h\in\cC_1\cup\cC_2\cup\cC_3\cup\cC_4\cup\cC_5\). Now note that we must have \(x_e = 0\), \(\hat{x}_e = \frac12\) and \(x^*_e = 1\). So \(g\) creates at least one cycle, and \(h\) breaks at least one cycle. This gives the following options for \(g\):
    \begin{itemize}
        \item \(g \in \cC_1 \cup \cC_5\) both breaks and creates a cycle,
        \item \(g \in \cC_2 \cup \cC_4\) creates one cycle,
        \item \(g \in \cC_5\) creates two cycles.
    \end{itemize}
    Hence, \(\Abs{\oddcyclesx{\hat{x}}} = \Abs{\oddcyclesx{x}} + \CBra{0,1,2}\). Similarly for \(h\):
    \begin{itemize}
        \item \(h \in \cC_1 \cup \cC_5\) both breaks and creates a cycle,
        \item \(h \in \cC_2 \cup \cC_4\) breaks one cycle,
        \item \(h \in \cC_5\) breaks two cycles.
    \end{itemize}
    Hence, \(\Abs{\oddcyclesx{x^*}} = \Abs{\oddcyclesx{\hat{x}}} - \CBra{0,1,2}\). So definitely, \(\Abs{\oddcyclesx{x^*}} \leq \Abs{\oddcyclesx{x}} + 2\), and consequently, as before, \(\gamma(\graph) \leq \gamma(\graph \setminus e) + 2\), yielding \eqref{eq: lb gamma(G) removed edge}.    
\end{proof}

The authors of \cite{Koh2020Stabilizing} provided an example that shows this bound is tight, already in the unit-capacity setting. However, in the unit-weight, capacitated setting we can get a stronger bound. Repeating the proof above, but replacing the ``\(-2\)'' in \eqref{eq: lb gamma(G) removed edge} by ``\(-1\)'' and using that circuits in \(\cC_1 \cup \cC_5\) are not augmenting in cardinality, we can prove: 
\begin{lemma}
    In \(\graphunitc\), for every edge-stabilizer \(F\), \(\Abs{F} \geq \gamma(\graph)\).
\end{lemma}
The authors of \cite{Koh2020Stabilizing} also give a \(O(\Delta)\)-approximation algorithm, based on their algorithm for the vertex-stabilizer problem (instead of removing the vertices, all edges incident to those vertices are removed). Similarly, for capacitated instances we use \Cref{alg: capacitated vertex-stab by reducing capacity weighted}, and instead of reducing the capacity of the vertices, we remove all edges incident to those vertices, except the edges \(e\) such that \(e\in\matchedx{x}\). 

\begin{theorem}
    \label{thm: O(Delta) approx edge-stab}
    The edge-stabilizer problem admits an efficient \(O(\Delta)\)-ap\-prox\-i\-ma\-tion algorithm. 
\end{theorem}
\begin{proof}
    Let \(S=\CBra{v_1,\ldots,v_{\gamma(\graph)}}\) be the set of vertices found by \Cref{alg: capacitated vertex-stab by reducing capacity weighted}. Set \(F=\CBra{\delta(v)\setminus\matchedx{x} : v\in S}\). The size of \(F\) is at most \(\Delta\gamma(\graph)\). We claim that \(\graph\setminus F\) is stable, hence this gives us an \(O(\Delta)\)-approximation, by \Cref{lem: size of edge-stab}.

    Let \(\hat{x}\) be obtained from \(x\) by alternate rounding \(C_i\) exposing \(v_i\), for all \(i\in\CBra{1,\ldots,\gamma(\graph)}\). Note that \(\hat{x}\) is a basic fractional \(\capacity\)-matching in \(\graph\setminus F\) with \(\Abs{\oddcyclesx{\hat{x}}}=0\). Let \((\hat{y},\hat{z})\) be obtained from \((y,z)\) as follows:
    \begin{equation*}
        \hat{y}_v = \begin{cases}
            y_v & \text{if } v\notin S, \\
            0 & \text{if } v\in S,
        \end{cases}
        \quad
        \hat{z}_{uv} = \begin{cases}
            z_{uv} + y_u + y_v & \text{if } u,v\in S, \shortEdge{u}{v}\in\edgeSet\setminus F, \\
            z_{uv} + y_u & \text{if } u\in S, v\notin S, \shortEdge{u}{v}\in\edgeSet\setminus F, \\
            z_{uv} + y_v & \text{if } u\notin S, v\in S, \shortEdge{u}{v}\in\edgeSet\setminus F, \\
            z_{uv} & \text{if } u,v\notin S, \shortEdge{u}{v}\in\edgeSet\setminus F. \\
        \end{cases}
    \end{equation*}
    One can check that \(\hat{x}\) and \((\hat{y},\hat{z})\) satisfy complementary slackness in \(\graph\setminus F\).
    Consequently, \(\hat{x}\) is a basic maximum-weight fractional \(\capacity\)-matching in \(\graph\setminus F\) with \(\Abs{\oddcyclesx{\hat{x}}}=0\), and so \(\graph\setminus F\) is stable.
\end{proof}

If we restrict ourselves to unit-weight instances, like for capacity-stabilizers, we can show that any \emph{inclusion-wise minimal} edge-stabilizer preserves the size of a maximum-cardinality \(\capacity\)-matching.

\begin{theorem}
    \label{thm: min edge-stabilizer size preserving}
    In \(\graphunitc\), for any inclusion-wise minimal edge-stabilizer \(F\): \(\nu^\capacity(\graph \setminus F) = \nucG\).
\end{theorem}
\begin{proof}
    Let \(\matching\) be a maximum-cardinality \(\capacity\)-matching in \(\graph\) such that the overlap with \(F\) is minimum, i.e., \(\Abs{\matching \cap F}\) is minimum. Suppose for the sake of contradiction that \(\Abs{\matching \cap F} > 0\).

    Consider the graph \(\graph \setminus \Par{F \setminus \matching}\). Since \(\matching\) is avoided, \(\matching\) is a maximum-cardinality \(\capacity\)-matching in \(\graph \setminus \Par{F \setminus \matching}\).
    Let \(x\) be the indicator vector of \(\matching\), then by \Cref{thm: basic fractional c-matching}, \(x\) is basic.
    Since \(F\) is inclusion-wise minimal, \(\graph \setminus \Par{F \setminus \matching}\) is not stable, and thus \(x\) is not a maximum fractional \(\capacity\)-matching in \(\graph \setminus \Par{F \setminus \matching}\). So there must be a vertex \(x^*\) of \(\Pfcm\), adjacent to \(x\), with \(1^\top x^* > 1^\top x\). By \Cref{thm: adjacent vertices polytope fcm}, \(x^* = x + \alpha g\), where \(\alpha\in\CBra{\frac{1}{2},1}\) and \(g\in\cC_1\cup\cC_2\cup\cC_3\cup\cC_4\cup\cC_5\). Since \(\matching\) is maximum, \(x^*\) cannot be integral, so \(\alpha = \frac12\), and consequently by \Cref{thm: adjacent vertices polytope fcm}, \(g\notin\cC_3\). Furthermore, the circuits in \(\cC_1 \cup \cC_5\) are not augmenting in cardinality, so \(g \in \cC_2\cup\cC_4\). Let \(v \in \vertexSet\) be the only vertex with \(g(\delta(v)) \neq 0\). The circuits in \(\cC_2\cup\cC_4\) are only augmenting in cardinality if \(g(\delta(v)) > 0\). Consequently, \(x^*(\delta(v)) = x(\delta(v)) + 1 = d_v^{\matching} + 1\). Clearly, we have \(x^*(\delta(v)) \leq \capacity_v\), and so, \(d_v^\matching < \capacity_v\).

    The support of \(g\) consists of a (possibly empty) path \(P_g\) and an odd cycle \(C_g\), intersecting at only one vertex, such that the sign of \(g_e\) alternates. Since it is feasible to apply \(g\) to \(x\), it must be that if \(sgn(g_e) = -\), then \(x_e = 1\) (\(e \in \matching\)), and if \(sgn(g_e) = +\), then \(x_e = 0\) (\(e \notin \matching\)). Since \(g(\delta(v)) > 0\), the first and last edge of \(g\) satisfy \(sgn(g_e) = +\). Recall that \(d_v^\matching < \capacity_v\). Consider the closed \(vv\)-walk \(W = (P_g, C_g, P_g^{-1})\) (\(P_g^{-1}\) is \(P_g\) backwards). Then \(W\) is a feasible \(\matching\)-augmenting walk in \(\graph \setminus ( F \setminus \matching )\).

    \begin{claim}
        \label{claim: intersection W F}
        \(W \cap F = \emptyset\).
    \end{claim}
    \begin{proof}
        For the sake of contradiction, suppose that there is some \(e \in W \cap F\). Then \(e \in \matching\), since all other edges of \(F\) were removed. Then there exists an even-length \(\matching\)-alternating trail \(T\): the part of \(W\) starting from \(v\), to (and including) the edge \(e\). Since \(d_v^\matching < \capacity_v\) and \(e \in \matching\), \(T\) is proper. Then \(\matching' = \matching \symdif T\) is a maximum-cardinality \(\capacity\)-matching in \(\graph\) with \(\Abs{\matching' \cap F} < \Abs{\matching \cap F}\), contradicting our assumption.
    \end{proof}
    
    By this claim, \(W\) exist in \(\graph \setminus F\). In addition, we have \(d_v^{\matching \setminus F} \leq d_v^\matching < \capacity_v\). Thus, \(W\) is a feasible \(\matching \setminus F\)-augmenting walk in \(\graph \setminus F\). By \Cref{thm: feasible augmenting walk not stable}, since \(\graph \setminus F\) is stable, it follows that \(\matching \setminus F\) is not maximum in \(\graph \setminus F\). Then, by \Cref{thm: M max iff no augmenting trail}, there is a proper \(\matching \setminus F\)-augmenting \(st\)-trail \(T\) in \(\graph \setminus F\). Note that since \(T\) is augmenting in cardinality, the first and last edge of \(T\) are not in \(\matching \setminus F\). Since \(W\) cannot exists for a maximum-cardinality \(\capacity\)-matching in \(\graph \setminus F\), by \Cref{thm: feasible augmenting walk not stable}, there must be such a trail \(T\) that either makes \(W\) infeasible, or overlaps with the edges of \(W\). Note that \(T\) is also an \(\matching\)-augmenting trail in \(\graph\), but \(\matching\) is maximum in \(\graph\), so by \Cref{thm: M max iff no augmenting trail}, \(T\) is not proper for \(\matching\) in \(\graph\).

    Since \(T\) is proper for \(\matching \setminus F\) in \(\graph \setminus F\), we have, if \(s \neq t\), \(d_s^{\matching \setminus F} \leq \capacity_s - 1\) and \(d_t^{\matching \setminus F} \leq \capacity_t - 1\), if instead \(s = t\), then \(d_s^{\matching \setminus F} \leq \capacity_s - 2\). 
    Since \(T\) is not proper for \(\matching\) in \(\graph\), we have, if \(s \neq t\), \(d_s^\matching = \capacity_s\) or \(d_t^\matching = \capacity_t\), if instead \(s = t\), then \(d_s^\matching \geq \capacity_s - 1\).
    This gives us five cases:
    \begin{enumerate}
        \item \(s \neq t\), \(d_s^\matching = \capacity_s\) and \(d_t^\matching < \capacity_t\),
        \item \(s \neq t\), \(d_s^\matching < \capacity_s\) and \(d_t^\matching = \capacity_t\),
        \item \(s \neq t\), \(d_s^\matching = \capacity_s\) and \(d_t^\matching = \capacity_t\),
        \item \(s = t\) and \(d_s^\matching = \capacity_s - 1\),
        \item \(s = t\) and \(d_s^\matching = \capacity_s\).
    \end{enumerate}
    In the first case, since \(d_s^\matching = \capacity_s\) and \(d_s^{\matching \setminus F} \leq \capacity_s - 1\), there is at least one \(e \in \delta(s) \cap \matching \cap F\). Since \(e \in F\), we have \(e \notin T\), which means \(T \cup e\) is a trail. Note that \(T \cup e\) is \(\matching\)-alternating, and has even length. Since \(d_t^\matching < \capacity_t\) and \(e \in \matching\), \(T \cup e\) is proper. (If \(T \cup e\) is closed, it is still proper, because it has even-length.) Therefore, \(\matching' = \matching \symdif \Par{T \cup e}\) is a maximum-cardinality \(\capacity\)-matching in \(\graph\) with \(\Abs{\matching' \cap F} < \Abs{\matching \cap F}\), contradicting our assumption. Similar arguments can be made in the second and fourth case. 

    For the third and fifth case we take a look at the overlap of \(T\) and \(W\). We know that \(d_v^{\matching \setminus F} \leq d_v^\matching < \capacity_v\). So, to make \(W\) infeasible in \(\graph \setminus F\), \(T\) would need to increase the degree of \(v\) with respect to the matching. This is only possible if \(v \in \CBra{s,t}\). However, \(s\) and \(t\) are both saturated by \(\matching\) (in both the third and fifth case) and \(v\) is not. Hence, \(T\) and \(W\) must overlap in at least one edge.

    We create a new walk \(W'\) by combining \(W\) and \(T\): follow \(W\) starting from \(v\) to the first edge \(e\) that is also on \(T\), traverse \(e\), then switch to \(T\) and follow \(T\) from here to one of its endpoints. Since \(T\) and \(W\) are both \(\matching\)-alternating and -augmenting, so is \(W'\). 
    The part of \(W\) from \(v\) to and including \(e\) is a trail, this part of \(W\) does not overlap with \(T\), and \(T\) is a trail, hence \(W'\) is a trail. Since \(T\) exists in \(\graph \setminus F\), and \(W \cap F = \emptyset\) by \Cref{claim: intersection W F}, we have \(W' \cap F = \emptyset\).
    
    In both the third and fifth case, the degree of both \(s\) and \(t\) with respect to \(\matching\) is strictly larger than their degree with respect to \(\matching \setminus F\). Hence, \(\delta(s) \cap \matching \cap F\) and \(\delta(t) \cap \matching \cap F\) are nonempty. Without loss of generality assume \(W'\) ends at \(s\). Let \(e \in \delta(s) \cap \matching \cap F\). Since \(e \in F\), we have \(e \notin W'\), which means \(W' \cup e\) is a trail. Note that \(W' \cup e\) is \(\matching\)-alternating, and has even length. Since \(d_v^\matching < \capacity_v\) and \(e \in \matching\), \(W' \cup e\) is proper. Therefore, \(\matching' = \matching \symdif \Par{W' \cup e}\) is a maximum-cardinality \(\capacity\)-matching in \(\graph\) with \(\Abs{\matching' \cap F} < \Abs{\matching \cap F}\), contradicting our assumption.
\end{proof}

We conclude this paper with some additional remarks. Note that, as opposed to the capacity-stabilizer case, when dealing with edge-removal operations it is always possible to stabilize a graph without decreasing the weight of a maximum-weight matching: for example, one could take any maximum-weight \(\capacity\)-matching \(\matching\) in \(\graph\) and remove all edges in \(\edgeSet \setminus \matching\). The previous theorem shows that, for unit-weight instances, this property comes essentially for free, as any edge-stabilizer of minimum cardinality will be weight-preserving. However, for general weighted instances, this is not the case and we can show that the size of a minimum weight-preserving edge-stabilizer and the size of a minimum edge-stabilizer, can differ by a very large factor, namely \(\Omega(\Abs{\vertexSet})\), already for unit-capacities.
\begin{theorem}
    There exist graphs \(\graphwunit\) where the sizes of a minimum edge-stabilizer and a minimum weight-preserving edge-stabilizer differ by \(\Omega(\Abs{\vertexSet})\).
\end{theorem}
\begin{proof}
    Let \(\graph\) and \(\matching\) be the graph and matching given in \Cref{fig: size of edge-stabilizers difference}, respectively.
    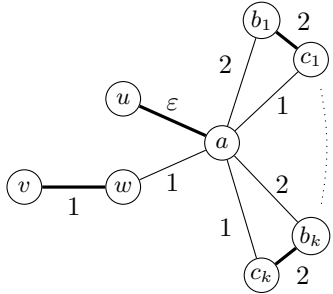
\begin{figure}[ht]
        \centering
        \begin{tikzpicture}[
            scale=1.3,
            circ/.style={
                circle,
                draw=black,
                inner sep=1pt,
                minimum size=13pt,
                font=\small
            },
            matched/.style={very thick}
        ]       
            \node[circ] (u) at (0,-.3) {\(u\)};
            \node[circ] (v) at (-1,-1.2) {\(v\)};
            \node[circ] (w) at (0,-1.2) {\(w\)};
            \node[circ] (a) at (1,-.75) {\(a\)};
            \node[circ] (b1) at (1.4,0.5) {\(b_1\)};
            \node[circ] (c1) at (1.9,0.1) {\(c_1\)};
            \node[circ] (bk) at (1.9,-1.7) {\(b_k\)};
            \node[circ] (ck) at (1.4,-2.1) {\(c_k\)};
        
            \draw[matched] (u) to node [above] {\(\varepsilon\)} (a);
            \draw[matched] (v) to node [below] {1} (w);
            \draw (w) to node [below] {1} (a);
            \draw (a) to node [anchor=south east] {2} (b1);
            \draw (a) to node [anchor=north west,yshift=5pt] {1} (c1);
            \draw[matched] (b1) to node [anchor=south west] {2} (c1);
            \draw (a) to node [anchor=south west,yshift=-5pt] {2} (bk);
            \draw (a) to node [anchor=north east] {1} (ck);
            \draw[matched] (bk) to node [anchor=north west] {2} (ck);
            \draw[dotted] (2,-.2) to [out=-80,in=80] (2,-1.4);
        \end{tikzpicture}
        \caption{Let \(k \geq 3\) be an integer, and \(0 < \varepsilon < 0.5\). The figure shows a graph \(\graphwunit\) with \(k\) \(3\)-cycles of the form \((a,b_i,c_i)\) for \(i=1,\ldots,k\). Edge weights are given next to the edges. The bold edges indicate a matching \(\matching\).}
        \label{fig: size of edge-stabilizers difference}
    \end{figure}
    Any matching on \(\graph\) can contain (i) at most one edge of each \(3\)-cycle, which all have a weight of at most \(2\), (ii) at most one of the edges incident to \(w\), which both have weight \(1\), and (iii) the edge \(\shortEdge{u}{a}\) of weight \(\varepsilon\). This gives us \(\nucG \leq 2k + 1 + \varepsilon\). In particular, \(\matching\) is the only matching in \(\graph\) that attains this weight. So, \(\matching\) is the unique maximum-weight matching in \(\graph\). There are \(k\) feasible \(\matching\)-augmenting walks: \((u, a, b_i, c_i,a,u)\) for \(i = 1, \ldots, k\), hence \(\graph\) is not stable, by \Cref{thm: feasible augmenting walk not stable}.

    We can stabilize \(\graph\) by removing only two edges: \(\shortEdge{u}{a}\) and \(\shortEdge{v}{w}\). We verify the stability by giving a matching and a vertex-cover of the same value. Let the matching be \(\shortEdge{w}{a}\cup\CBra{\shortEdge{b_i}{c_i} : i\in\CBra{1,\ldots,k}}\). The weight of this matching is \(2k+1\). Note that the weight decreased by \(\varepsilon\). We set the vertex-cover \(y\) equal to \(0\) for \(u,v,w\) and equal to \(1\) for \(a\) and all \(k\) pairs \(b_i,c_i\). The value of this vertex-cover is \(2k+1\). So indeed, the graph is stable. This shows that the size of a minimum edge-stabilizer is at most two.

    As mentioned, \(\matching\) is the unique maximum-weight matching in \(\graph\). Consequently, any weight-preserving edge-stabilizer has to avoid \(\matching\). Because the \(k\) feasible \(\matching\)-augmenting walks are edge-disjoint with respect to the edges from \(\edgeSet\setminus\matching\), any weight-preserving edge-stabilizer has to remove at least \(k\) edges. 
    
    We have \(\Abs{\vertexSet}=3k+3\), or equivalently, \(k = \Abs{\vertexSet}/3 - 1\). The difference in size of a minimum edge-stabilizer and a minimum weight-preserving edge-stabilizer in this graph is at least \(k-2\). So the difference in sizes is \(\Omega(k)=\Omega(\Abs{\vertexSet})\).  
\end{proof}

Theorem 9 of \cite{Koh2020Stabilizing} shows that there is no constant factor approximation for the minimum edge-stabilizer problem in \(\graphwunit\), unless \(\Po=\NP\). We note that their proof actually also shows inapproximability for the minimum weight-preserving edge-stabilizer problem.

Our last remark is the following. For the vertex-stabilizer problem, the setting of removing the vertices completely has been analyzed in~\cite{gerstbrein2022stabilization}, and that of reducing the capacity has been investigated in this paper.  
One may wonder whether a similar setting of ``partial reduction'' also makes sense for the edge-stabilizer problem: what if one reduces the weight of the edges, instead of completely removing them? 
The next theorem suggests that reducing edge weights might not be that interesting from a bargaining perspective: if the weight of an edge is decreased, it will not be part of any maximum-weight \(\capacity\)-matching, so one just as well could have removed the edge.

\begin{theorem}
    Let \(\graphwc\) be a graph, and let \(\widetilde{\weight} \in \R_{\geq0}^\edgeSet\) such that \((\graph,\weight-\widetilde{\weight},\capacity)\) is stable and \(1^\top \widetilde{\weight}\) is minimum. Let \(f \in \edgeSet\) such that \(\widetilde{\weight}_f > 0\), then \(f\) is not part of any maximum-weight \(\capacity\)-matching in \((\graph,\weight-\widetilde{\weight},\capacity)\).
\end{theorem}
\begin{proof}
    For the sake of contradiction let \(\matching\) be a maximum-weight \(\capacity\)-matching in \((\graph,\weight-\widetilde{\weight},\capacity)\) with \(f \in \matching\). Let \(x\) be the indicator vector of \(\matching\). Since \((\graph,\weight-\widetilde{\weight},\capacity)\) is stable, there is a fractional vertex cover \((y,z)\) that satisfies complementary slackness with \(x\):
    \begin{align*}
        x_e = 0 &\lor y_u+y_v+z_e=\weight_e-\widetilde{\weight}_e &\forall e=uv \in \edgeSet \\
        y_v = 0 &\lor x(\delta(v)) = \capacity_v &\forall v \in \vertexSet \\
        z_e = 0 &\lor x_e = 1 &\forall e \in \edgeSet
    \end{align*}
    Now let \(x'=x\), \(y'=y\), \(z'_f=z_f+\widetilde{\weight}_f\), \(z'_e=z_e\) for all \(e \neq f\), \(\weight'_f = 0\), and \(\weight'_e=\widetilde{\weight}_e\) for all \(e \neq f\). One can check that \(x'\) and \((y',z')\) are a \(\capacity\)-matching and fractional vertex cover in \((\graph,\weight-\weight',\capacity)\), respectively, and that they satisfy complementary slackness. Hence, \((\graph,\weight-\weight',\capacity)\) is stable, but \(1^\top \weight' < 1^\top\widetilde{\weight}\), contradicting that \(1^\top \widetilde{\weight}\) is minimum.
    
    In case of unit-capacities there are no \(z\) variables in the dual. To make the proof still valid, it is enough to increase the \(y\)-value of one of the vertices of \(f\) with \(\widetilde{\weight}_f\).
\end{proof}

\appendix
\section{Reduction to Unit-capacity Instances}\label{appx: reduction}


One can transform a capacitated graph \(\graphwc\) to a unit-capacity graph by replacing each vertex \(v\) by \(\capacity_v\) copies, and each edge \(uv\) by the complete bipartite graph between \(u_1,\ldots, u_{\capacity_u}\) and \(v_1,\ldots,v_{\capacity_v}\). However, this can transform unstable graphs to stable graphs. 
For example, the graph in \Cref{subfig: simple reduction capacitated} is unstable, but the unit-capacity graph that we can obtain from it using the described transformation, given in \Cref{subfig: simple reduction unit-capacity}, is stable. This is due to the edges between the vertex copies of \(v\) and \(x\): in the capacitated instance the edge \(vx\) can be matched at most once, but in the unit-capacity instance two copies of the edge can be matched, e.g., \(v_1 x_1\) and \(v_2 x_2\).
\begin{figure}[ht]
    \centering
    \begin{subfigure}{.45\textwidth}
        \centering
        \begin{tikzpicture}[
            scale=.6,
            circ/.style={
                circle,
                draw=black,
                inner sep=2pt
            }
        ]       
            \node[circ,label={above:\(t\)}] (u) at (0,0) {};
            \node[circ,label={above:\(u\)}] (v) at (2,0) {};
            \node[circ,label={above:\(v\)},label={below left:\(2\)}] (w) at (4,0) {};
            \node[circ,label={above:\(x\)},label={below:\(2\)}] (x) at (6,0) {};
            \node[circ,label={above:\(y\)}] (y) at (7.5,1) {};
            \node[circ,label={below:\(z\)}] (z) at (7.5,-1) {};
            \node[circ,label={below:\(c\)}] (c) at (4,-2) {};
            \node[circ,label={below:\(b\)},label={above:\(2\)}] (b) at (2,-2) {};
            \node[circ,label={below:\(a\)}] (a) at (0,-2) {};
        
            \draw (u) to (v);
            \draw (v) to (w);
            \draw (w) to (x);
            \draw (x) to (y);
            \draw (x) to (z);
            \draw (y) to (z);
            \draw (a) to (b);
            \draw (b) to (c);
            \draw (c) to (w);
        \end{tikzpicture}
        \caption{Graph \(\graphunitc\) with \(\capacity_v=\capacity_x=\capacity_b = 2\) and \(\capacity=1\) for the other vertices.}
        \label{subfig: simple reduction capacitated}
    \end{subfigure}
    \begin{subfigure}{.45\textwidth}
        \centering
        \begin{tikzpicture}[
            scale=.6,
            circ/.style={
                circle,
                draw=black,
                inner sep=2pt
            }
        ]       
            \node[circ,label={above:\(t_1\)}] (u) at (0,0) {};
            \node[circ,label={above:\(u_1\)}] (v) at (2,0) {};
            \node[circ,label={above:\(v_1\)}] (w1) at (4,.5) {};
            \node[circ,label={left: \(v_2\)}] (w2) at (4,-.5) {};
            \node[circ,label={above:\(x_1\)}] (x1) at (6,.5) {};
            \node[circ,label={below:\(x_2\)}] (x2) at (6,-.5) {};
            \node[circ,label={above:\(y_1\)}] (y) at (7.5,1) {};
            \node[circ,label={below:\(z_1\)}] (z) at (7.5,-1) {};
            \node[circ,label={below:\(c_1\)}] (c) at (4,-2) {};
            \node[circ,label={above:\(b_1\)}] (b1) at (2,-1.5) {};
            \node[circ,label={below:\(b_2\)}] (b2) at (2,-2.5) {};
            \node[circ,label={below:\(a_1\)}] (a) at (0,-2) {};
        
            \draw (u)  to (v);
            \draw (v)  to (w1);
            \draw (v)  to (w2);
            \draw (w1) to (x1);
            \draw (w1) to (x2);
            \draw (w2) to (x1);
            \draw (w2) to (x2);
            \draw (x1) to (y);
            \draw (x2) to (y);
            \draw (x1) to (z);
            \draw (x2) to (z);
            \draw (y)  to (z);
            \draw[bend left=30] (w1) to (c);
            \draw (w2) to (c);
            \draw (b1) to (c);
            \draw (b2) to (c);
            \draw (b1) to (a);
            \draw (b2) to (a);
        \end{tikzpicture}
        \caption{The corresponding unit-capacity graph.}
        \label{subfig: simple reduction unit-capacity}
    \end{subfigure}
    \caption{Example of reducing a capacitated graph to a unit-capacity graph, by replacing each vertex \(v\) by \(\capacity_v\) copies.}
    \label{fig: simple reduction}
\end{figure}
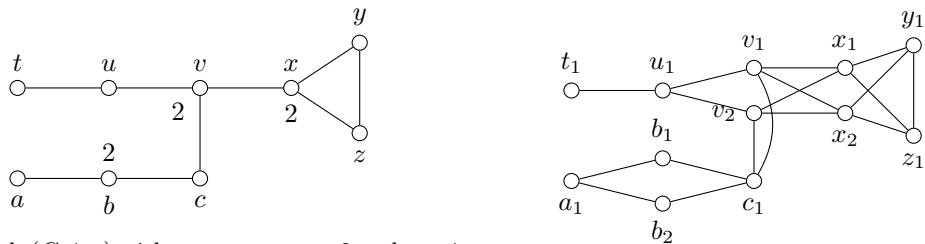


A more sophisticated way to transform a capacitated graph \(\graphwc\) into a unit-capacity graph is the following:
First we reduce the problem to a perfect \(\capacity\)-matching problem, by taking two copies of the graph, and adding \(\capacity_v\) parallel edges between the two copies of each vertex \(v\), of weight zero. Then, we replace each edge \(e=uv\) by two new vertices: \(e_u\), \(e_v\), and three new edges: \(u e_v\), \(e_v e_u\), \(e_u v\) of weight \(\frac12 \weight_e\), \(0\), and \(\frac12 \weight_e\), respectively. Finally, we replace each (original) vertex \(v\) by \(\capacity_v\) copies: \(v_1,\ldots,v_{\capacity_v}\), and each edge \(v e_u\) by an edge from each of the \(\capacity_v\) copies of \(v\) to \(e_u\), all of the same weight as \(v e_u\) (so \(\frac12 \weight_e\)).

The problems with this reduction are twofold. First, if we blindly apply the unit-capacity algorithm for the vertex-stabilizer problem, we might reduce the capacity of a vertex that does not have the lowest \(y\)-value of its fractional odd cycle. As a result more weight is lost then necessary, and the current proof of the weight-preserving factor of \(\frac23\) does not work anymore (because it depends on vertices with the lowest \(y\)-value being used). 
Second, edge-stabilizers are not preserved.
We demonstrate the first  problem using the following example.

Consider the graph \(\graphwc\) in \Cref{fig: reduction half w_e}. The edge weights and vertex capacities are given in \Cref{subfig: reduction half w_e weights,subfig: reduction half w_e capacities}, respectively. \Cref{subfig: reduction half w_e frac matching} displays a maximum-weight fractional \(\capacity\)-matching (with \(\Abs{\oddcyclesx{x}} = \gamma(\graph)\)) and minimum fractional vertex cover. 
\Cref{subfig: reduction half w_e auxiliary} gives part of the unit-capacity auxiliary graph, namely one of the graph copies. The full auxiliary graph is this graph twice, connected by edges consisting of three parts, where all three parts have weight zero, the vertices that split the parts have \(y\)-value zero, and such that the middle part of those edges is matched. The given fractional matching and vertex cover are optimal. Optimality of the fractional (\(\capacity\)-)matchings and vertex covers given in these figures can be verified by checking complementary slackness.
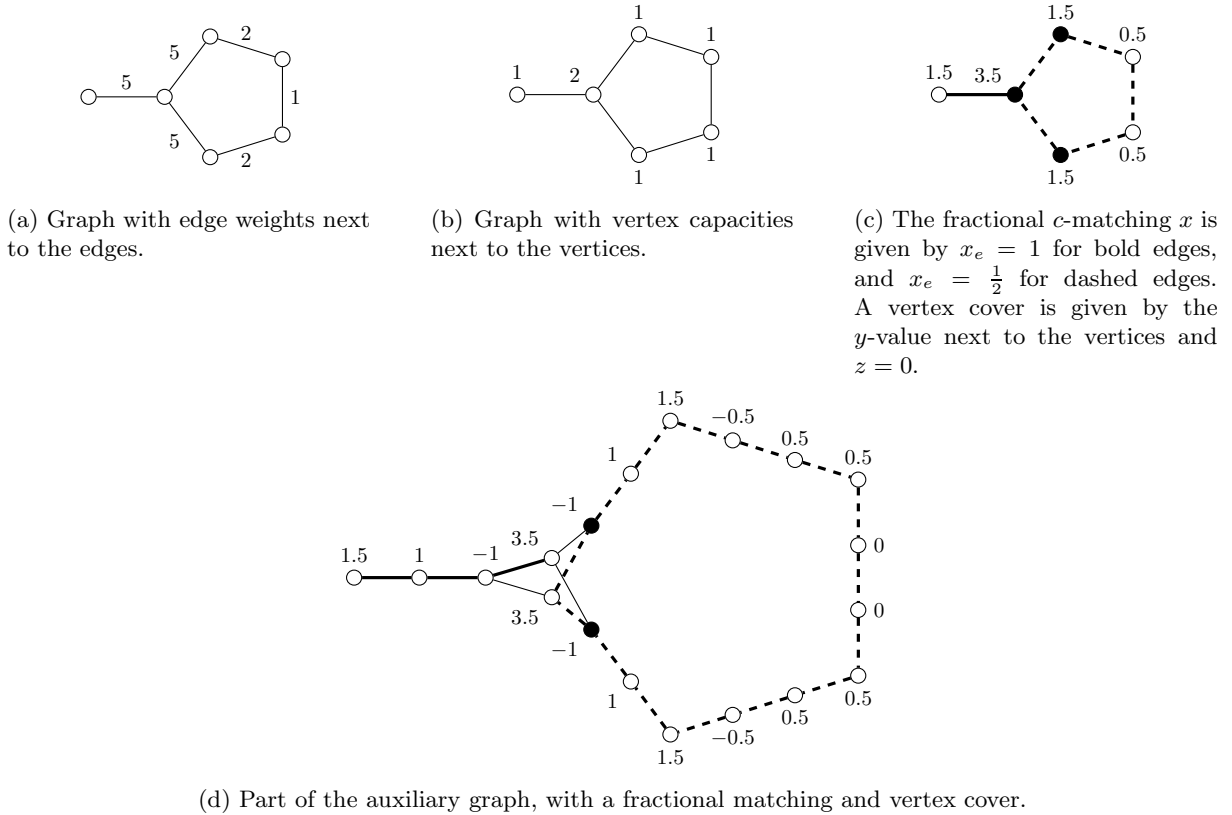
\begin{figure}[ht]
    \centering
    \begin{subfigure}[t]{.3\textwidth}
        \centering
        \begin{tikzpicture}[
            scale=.5,
            circ/.style={
                circle,
                draw=black,
                inner sep=2.5pt
            },
            every node/.style={scale=0.8}
        ]
            \node[circ] (a) at (0,0) {};
            \node[circ] (b) at (2,0) {};
            \node[circ] (c) at (3.2,1.6) {};
            \node[circ] (d) at (5.1,1) {};
            \node[circ] (e) at (5.1,-1) {};
            \node[circ] (f) at (3.2,-1.6) {};
    
            \draw (a) to node [above] {\(5\)} (b);
            \draw (b) to node [anchor=south east] {\(5\)} (c);
            \draw (c) to node [above] {\(2\)} (d);
            \draw (d) to node [right] {\(1\)} (e);
            \draw (e) to node [below] {\(2\)} (f);
            \draw (f) to node [anchor=north east] {\(5\)} (b);

            \path (-.5,2.5) -- (5.5,2.5) -- (5.5,-2.5) -- (-.5,-2.5) -- (-.5,2.5);
        \end{tikzpicture}
        \caption{Graph with edge weights next to the edges.}
        \label{subfig: reduction half w_e weights}
    \end{subfigure}
    \hfill
    \begin{subfigure}[t]{.3\textwidth}
        \centering
        \begin{tikzpicture}[
            scale=.5,
            circ/.style={
                circle,
                draw=black,
                inner sep=2.5pt
            },
            every node/.style={scale=0.8}
        ]
            \node[circ,label={above:\(1\)}] (a) at (0,0) {};
            \node[circ,label={north west:\(2\)}] (b) at (2,0) {};
            \node[circ,label={above:\(1\)}] (c) at (3.2,1.6) {};
            \node[circ,label={above:\(1\)}] (d) at (5.1,1) {};
            \node[circ,label={below:\(1\)}] (e) at (5.1,-1) {};
            \node[circ,label={below:\(1\)}] (f) at (3.2,-1.6) {};
    
            \draw (a) to (b);
            \draw (b) to (c);
            \draw (c) to (d);
            \draw (d) to (e);
            \draw (e) to (f);
            \draw (f) to (b);

            \path (-.5,2.5) -- (5.5,2.5) -- (5.5,-2.5) -- (-.5,-2.5) -- (-.5,2.5);
        \end{tikzpicture}
        \caption{Graph with vertex capacities next to the vertices.}
        \label{subfig: reduction half w_e capacities}
    \end{subfigure}
    \hfill
    \begin{subfigure}[t]{.3\textwidth}
        \centering
        \begin{tikzpicture}[
            scale=.5,
            circ/.style={
                circle,
                draw=black,
                inner sep=2.5pt
            },
            matched/.style={very thick},
            fracMatched/.style={dashed,very thick},
            every node/.style={scale=0.8}
        ]
            \node[circ,label={above:\(1.5\)}] (a) at (0,0) {};
            \node[circ,fill,label={north west:\(3.5\)}] (b) at (2,0) {};
            \node[circ,fill,label={above:\(1.5\)}] (c) at (3.2,1.6) {};
            \node[circ,label={above:\(0.5\)}] (d) at (5.1,1) {};
            \node[circ,label={below:\(0.5\)}] (e) at (5.1,-1) {};
            \node[circ,fill,label={below:\(1.5\)}] (f) at (3.2,-1.6) {};
    
            \draw[matched] (a) to (b);
            \draw[fracMatched] (b) to (c);
            \draw[fracMatched] (c) to (d);
            \draw[fracMatched] (d) to (e);
            \draw[fracMatched] (e) to (f);
            \draw[fracMatched] (f) to (b);

            \path (-.5,2.5) -- (5.5,2.5) -- (5.5,-2.5) -- (-.5,-2.5) -- (-.5,2.5);
        \end{tikzpicture}
        \caption{The fractional \(\capacity\)-matching \(x\) is given by \(x_e=1\) for bold edges, and \(x_e=\frac12\) for dashed edges. A vertex cover is given by the \(y\)-value next to the vertices and \(z=0\).}
        \label{subfig: reduction half w_e frac matching}
    \end{subfigure}
    \begin{subfigure}[t]{\textwidth}
        \centering
        \begin{tikzpicture}[
            scale=1.3,
            circ/.style={
                circle,
                draw=black,
                inner sep=2.5pt
            },
            matched/.style={very thick},
            fracMatched/.style={dashed,very thick},
            every node/.style={scale=0.8}
        ]
            \node[circ,label={above:\(1.5\)}] (a) at (0,0) {};
            \node[circ,label={above:\(1\)}] (a_b) at (.66,0) {};
            \node[circ,label={above:\(-1\)}] (b_a) at (1.33,0) {};
            \node[circ,label={north west:\(3.5\)}] (b1) at (2,.2) {};
            \node[circ,label={south west:\(3.5\)}] (b2) at (2,-.2) {};
            \node[circ,label={north west:\(-1\)},fill] (b_c) at (2.4,0.53) {};
            \node[circ,label={north west:\(1\)}] (c_b) at (2.8,1.06) {};
            \node[circ,label={above:\(1.5\)}] (c) at (3.2,1.6) {};
            \node[circ,label={above:\(-0.5\)}] (c_d) at (3.83,1.4) {};
            \node[circ,label={above:\(0.5\)}] (d_c) at (4.46,1.2) {};
            \node[circ,label={above:\(0.5\)}] (d) at (5.1,1) {};
            \node[circ,label={right:\(0\)}] (d_e) at (5.1,0.33) {};
            \node[circ,label={right:\(0\)}] (e_d) at (5.1,-0.33) {};
            \node[circ,label={below:\(0.5\)}] (e) at (5.1,-1) {};
            \node[circ,label={below:\(0.5\)}] (e_f) at (4.46,-1.2) {};
            \node[circ,label={below:\(-0.5\)}] (f_e) at (3.83,-1.4) {};
            \node[circ,label={below:\(1.5\)}] (f) at (3.2,-1.6) {};
            \node[circ,label={south west:\(1\)}] (f_b) at (2.8,-1.06) {};
            \node[circ,label={south west:\(-1\)},fill] (b_f) at (2.4,-0.53) {};
    
            \draw[matched] (a) to (a_b);
            \draw[matched] (a_b) to (b_a);
            \draw[matched] (b_a) to (b1);
            \draw (b_a) to (b2);
            \draw (b1) to (b_c);
            \draw[fracMatched] (b2) to (b_c);
            \draw[fracMatched] (b_c) to (c_b);
            \draw[fracMatched] (c_b) to (c);
            \draw[fracMatched] (c) to (c_d);
            \draw[fracMatched] (c_d) to (d_c);
            \draw[fracMatched] (d_c) to (d);
            \draw[fracMatched] (d) to (d_e);
            \draw[fracMatched] (d_e) to (e_d);
            \draw[fracMatched] (e_d) to (e);
            \draw[fracMatched] (e) to (e_f);
            \draw[fracMatched] (e_f) to (f_e);
            \draw[fracMatched] (f_e) to (f);
            \draw[fracMatched] (f) to (f_b);
            \draw[fracMatched] (f_b) to (b_f);
            \draw (b_f) to (b1);
            \draw[fracMatched] (b_f) to (b2);
        \end{tikzpicture}
        \caption{Part of the auxiliary graph, with a fractional matching and vertex cover.}
        \label{subfig: reduction half w_e auxiliary}
    \end{subfigure}
    \caption{Example of second reduction.}
    \label{fig: reduction half w_e}
\end{figure}

If we blindly apply the unit-capacity algorithm to the auxiliary graph, we would remove one vertex per fractional odd cycle, and in particular the one with the lowest \(y\)-value. The two black vertices in \Cref{subfig: reduction half w_e auxiliary} are the potential vertices to be removed (in both graph copies). Mapping this vertex-stabilizer to a capacity-stabilizer in \(\graph\), we would reduce the capacity of one of the black vertices in \Cref{subfig: reduction half w_e frac matching}. None of these vertices have the lowest \(y\)-value on the fractional odd cycle. 
Let \(\widetilde{\graph}\) be the stabilized (capacitated) graph. Then \(\nucG=12\) and \(\nu^\capacity(\widetilde{\graph}) = 11\) or \(9\). Both possibilities satisfy \(\nu^\capacity(\widetilde{\graph}) \geq \frac23 \nucG = 8\), but they do reduce the weight, while removing one of the white vertices (on the fractional odd cycle) would result in a stable graph \(\widehat{\graph}\) with \(\nu^\capacity(\widehat{\graph}) = 12 = \nucG\). This also holds true in general: the higher the \(y\)-value of the vertices in the capacity-stabilizer, the more weight is lost.


Another way to transform a capacitated graph \(\graphwc\) into a unit-capacity graph is the following \cite{SchrijverCO}: We replace each edge \(e=uv\) by two new vertices: \(e_u\), \(e_v\), and three new edges: \(u e_v\), \(e_v e_u\), \(e_u v\) all of weight \(\weight_e\). We replace each (original) vertex \(v\) by \(\capacity_v\) copies: \(v_1, \ldots, v_{\capacity_v}\), and each edge \(v e_u\) by an edge from each of the \(\capacity_v\) copies of \(v\) to \(e_u\), all of the same weight as \(v e_u\) (so \(\weight_e\)).
(Note that for this reduction it is not necessary to go to a perfect matching instance: if there is a matching \(\matching\) such that \(u e_v \in \matching\) and \(e_v e_u, e_u v \notin \matching\), then \(\matching' = \matching \setminus u e_v \cup e_v e_u\) has the same weight as \(\matching\). Hence, without loss of generality, we can assume that for each original edge \(e=uv\) we have either \(u e_v, e_u v \in \matching\) and \(e_v e_u \notin \matching\), or \(e_v e_u \in \matching\) and \(u e_v, e_u v \notin \matching\).)

The problem with this reduction is that an edge-stabilizer in the unit-capacity auxiliary graph might not map to an edge-stabilizer in the original graph. We demonstrate this problem using an example.

Consider the graph \(\graphunitc\) in \Cref{fig: reduction 3 times w_e}. The vertex capacities are given in \Cref{subfig: reduction 3 times w_e capacities}. \Cref{subfig: reduction 3 times w_e frac matching,subfig: reduction 3 times w_e edge removed,subfig: reduction 3 times w_e auxiliary,subfig: reduction 3 times w_e auxiliary edge removed} all give a maximum-weight fractional (\(\capacity\)-)matching (with \(\Abs{\oddcyclesx{x}} = \gamma(\graph)\)), and a minimum fractional vertex cover. Optimality can be verified by checking the complementary slackness conditions.
\begin{figure}[ht]
    \centering
    \begin{subfigure}[t]{.3\textwidth}
        \centering
        \begin{tikzpicture}[
            scale=.5,
            circ/.style={
                circle,
                draw=black,
                inner sep=2.5pt
            },
            every node/.style={scale=0.8}
        ]
            \node[circ,label={left:\(1\)}] (a) at (-1.73,1) {};
            \node[circ,label={left:\(1\)}] (b) at (-1.73,-1) {};
            \node[circ,label={above:\(2\)}] (c) at (0,0) {};
            \node[circ,label={above:\(1\)}] (d) at (2,0) {};
            \node[circ,label={right:\(1\)}] (e) at (3.73,1) {};
            \node[circ,label={right:\(1\)}] (f) at (3.73,-1) {};
    
            \draw (a) to (b);
            \draw (a) to (c);
            \draw (b) to (c);
            \draw (c) to (d);
            \draw (d) to (e);
            \draw (d) to (f);
            \draw (e) to (f);

            \path (-2.5,1.5) -- (4.5,1.5) -- (4.5,-1.5) -- (-2.5,-1.5) -- (-2.5,1.5);
        \end{tikzpicture}
        \caption{Unit-weight graph with vertex capacities next to the vertices.}
        \label{subfig: reduction 3 times w_e capacities}
    \end{subfigure}
    \hfill
    \begin{subfigure}[t]{.3\textwidth}
        \centering
        \begin{tikzpicture}[
            scale=.5,
            circ/.style={
                circle,
                draw=black,
                inner sep=2.5pt
            },
            matched/.style={very thick},
            fracMatched/.style={dashed,very thick},
            every node/.style={scale=0.8}
        ]
            \node[circ] (a) at (-1.73,1) {};
            \node[circ] (b) at (-1.73,-1) {};
            \node[circ] (c) at (0,0) {};
            \node[circ] (d) at (2,0) {};
            \node[circ] (e) at (3.73,1) {};
            \node[circ] (f) at (3.73,-1) {};
    
            \draw[fracMatched] (a) to (b);
            \draw[fracMatched] (a) to (c);
            \draw[fracMatched] (b) to (c);
            \draw[matched] (c) to (d);
            \draw (d) to (e);
            \draw (d) to (f);
            \draw[matched] (e) to (f);

            \path (-2.5,1.5) -- (4.5,1.5) -- (4.5,-1.5) -- (-2.5,-1.5) -- (-2.5,1.5);
        \end{tikzpicture}
        \caption{The fractional \(\capacity\)-matching \(x\) is given by \(x_e=1\) for bold edges, and \(x_e=\frac12\) for dashed edges. A vertex cover is given by \(y=\frac12\) and \(z=0\).}
        \label{subfig: reduction 3 times w_e frac matching}
    \end{subfigure}
    \hfill
    \begin{subfigure}[t]{.3\textwidth}
        \centering
        \begin{tikzpicture}[
            scale=.5,
            circ/.style={
                circle,
                draw=black,
                inner sep=2.5pt
            },
            matched/.style={very thick},
            fracMatched/.style={dashed,very thick},
            every node/.style={scale=0.8}
        ]
            \node[circ] (a) at (-1.73,1) {};
            \node[circ] (b) at (-1.73,-1) {};
            \node[circ] (c) at (0,0) {};
            \node[circ] (d) at (2,0) {};
            \node[circ] (e) at (3.73,1) {};
            \node[circ] (f) at (3.73,-1) {};
    
            \draw[matched] (a) to (c);
            \draw[matched] (b) to (c);
            \draw(c) to (d);
            \draw[fracMatched] (d) to (e);
            \draw[fracMatched] (d) to (f);
            \draw[fracMatched] (e) to (f);

            \path (-2.5,1.5) -- (4.5,1.5) -- (4.5,-1.5) -- (-2.5,-1.5) -- (-2.5,1.5);
        \end{tikzpicture}
        \caption{Same graph with one edge removed. Again a fractional \(\capacity\)-matching \(x\) is given by the bold and dashed edges, and a vertex cover by \(y=\frac12\) and \(z=0\).}
        \label{subfig: reduction 3 times w_e edge removed}
    \end{subfigure}
    \begin{subfigure}[t]{.49\textwidth}
        \centering
        \begin{tikzpicture}[
            scale=1.1,
            circ/.style={
                circle,
                draw=black,
                inner sep=2.5pt
            },
            matched/.style={very thick},
            fracMatched/.style={dashed,very thick},
            every node/.style={scale=0.8}
        ]
            \node[circ] (a) at (-1.73,1) {};
            \node[circ] (a_b) at (-1.73,.33) {};
            \node[circ] (b_a) at (-1.73,-.33) {};
            \node[circ] (a_c) at (-1.15,.66) {};
            \node[circ] (c_a) at (-0.58,.33) {};
            \node[circ] (b) at (-1.73,-1) {};
            \node[circ] (b_c) at (-1.15,-.66) {};
            \node[circ] (c_b) at (-0.58,-.33) {};
            \node[circ] (c1) at (0,.15) {};
            \node[circ] (c2) at (0,-.15) {};
            \node[circ] (c_d) at (.66,0) {};
            \node[circ] (d_c) at (1.33,0) {};
            \node[circ] (d) at (2,0) {};
            \node[circ] (d_e) at (2.58,.33) {};
            \node[circ] (e_d) at (3.15,.66) {};
            \node[circ] (d_f) at (2.58,-.33) {};
            \node[circ] (f_d) at (3.15,-.66) {};
            \node[circ] (e) at (3.73,1) {};
            \node[circ] (e_f) at (3.73,.33) {};
            \node[circ] (f_e) at (3.73,-.33) {};
            \node[circ] (f) at (3.73,-1) {};
    
            \draw[fracMatched] (a) to (a_b);
            \draw[fracMatched] (a_b) to (b_a);
            \draw[fracMatched] (b_a) to (b);
            \draw[fracMatched] (a) to (a_c);
            \draw[fracMatched] (a_c) to (c_a);
            \draw[fracMatched] (c_a) to (c1);
            \draw (c_a) to (c2);
            \draw[fracMatched] (b) to (b_c);
            \draw[fracMatched] (b_c) to (c_b);
            \draw[fracMatched] (c_b) to (c1);
            \draw (c_b) to (c2);
            \draw (c1) to (c_d);
            \draw[matched] (c2) to (c_d);
            \draw (c_d) to (d_c);
            \draw[matched] (d_c) to (d);
            \draw (d) to (d_e);
            \draw[matched] (d_e) to (e_d);
            \draw (e_d) to (e);
            \draw (d) to (d_f);
            \draw[matched] (d_f) to (f_d);
            \draw (f_d) to (f);
            \draw[matched] (e) to (e_f);
            \draw (e_f) to (f_e);
            \draw[matched] (f_e) to (f);
        \end{tikzpicture}
        \caption{The unit-capacity auxiliary graph, with a fractional matching given in the figure, and a vertex cover given by \(y=\frac12\).}
        \label{subfig: reduction 3 times w_e auxiliary}
    \end{subfigure}
    \hfill
    \begin{subfigure}[t]{.49\textwidth}
        \centering
        \begin{tikzpicture}[
            scale=1.1,
            circ/.style={
                circle,
                draw=black,
                inner sep=2.5pt
            },
            matched/.style={very thick},
            fracMatched/.style={dashed,very thick},
            every node/.style={scale=0.8}
        ]
            \node[circ] (a) at (-1.73,1) {};
            \node[circ,fill] (a_b) at (-1.73,.33) {};
            \node[circ,fill] (b_a) at (-1.73,-.33) {};
            \node[circ,fill] (a_c) at (-1.15,.66) {};
            \node[circ] (c_a) at (-0.58,.33) {};
            \node[circ] (b) at (-1.73,-1) {};
            \node[circ,fill] (b_c) at (-1.15,-.66) {};
            \node[circ] (c_b) at (-0.58,-.33) {};
            \node[circ,fill] (c1) at (0,.15) {};
            \node[circ,fill] (c2) at (0,-.15) {};
            \node[circ] (c_d) at (.66,0) {};
            \node[circ,fill] (d_c) at (1.33,0) {};
            \node[circ] (d) at (2,0) {};
            \node[circ,fill] (d_e) at (2.58,.33) {};
            \node[circ] (e_d) at (3.15,.66) {};
            \node[circ,fill] (d_f) at (2.58,-.33) {};
            \node[circ] (f_d) at (3.15,-.66) {};
            \node[circ,fill] (e) at (3.73,1) {};
            \node[circ] (e_f) at (3.73,.33) {};
            \node[circ] (f_e) at (3.73,-.33) {};
            \node[circ,fill] (f) at (3.73,-1) {};
    
            \draw[matched] (a) to (a_b);
            \draw[matched] (b_a) to (b);
            \draw (a) to (a_c);
            \draw[matched] (a_c) to (c_a);
            \draw (c_a) to (c1);
            \draw (c_a) to (c2);
            \draw (b) to (b_c);
            \draw[matched] (b_c) to (c_b);
            \draw (c_b) to (c1);
            \draw (c_b) to (c2);
            \draw (c1) to (c_d);
            \draw[matched] (c2) to (c_d);
            \draw (c_d) to (d_c);
            \draw[matched] (d_c) to (d);
            \draw (d) to (d_e);
            \draw[matched] (d_e) to (e_d);
            \draw (e_d) to (e);
            \draw (d) to (d_f);
            \draw[matched] (d_f) to (f_d);
            \draw (f_d) to (f);
            \draw[matched] (e) to (e_f);
            \draw (e_f) to (f_e);
            \draw[matched] (f_e) to (f);
        \end{tikzpicture}
        \caption{The unit-capacity auxiliary graph with one edge removed. Again a fractional matching is given in the figure, and a vertex cover is given by \(y_v=0\) for black vertices and \(y_v=1\) for white vertices.}
        \label{subfig: reduction 3 times w_e auxiliary edge removed}
    \end{subfigure}
    \caption{Example of third reduction.}
    \label{fig: reduction 3 times w_e}
\end{figure}
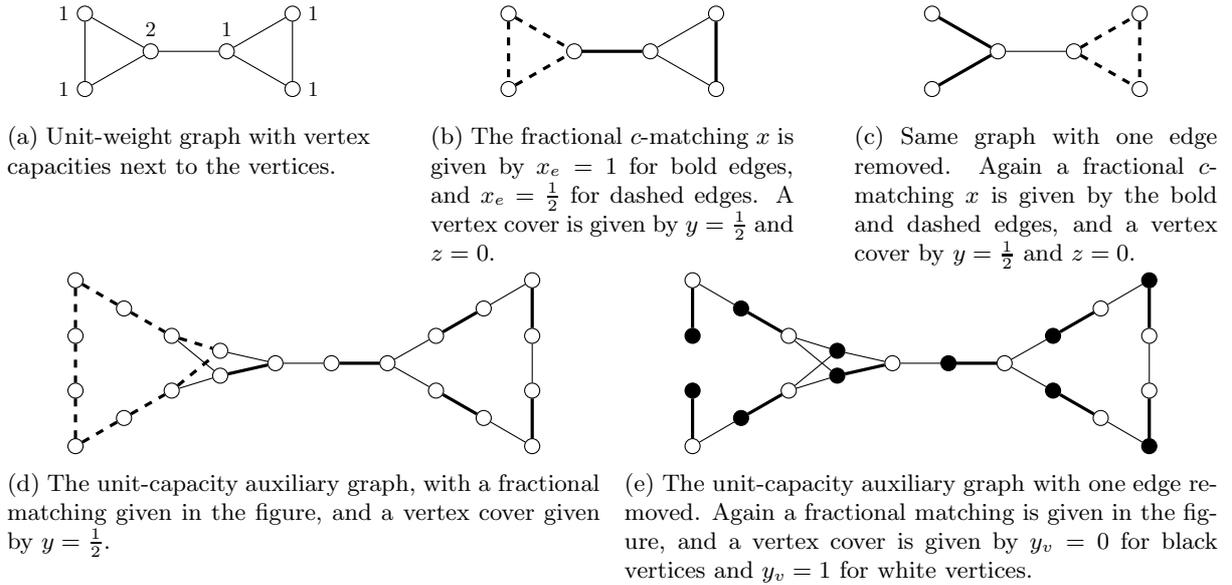

The unit-capacity auxiliary graph is given in \Cref{subfig: reduction 3 times w_e auxiliary}, and from the fractional matching we can see that the graph is not stable. The graph can be stabilized by removing just one edge, as can be seen in \Cref{subfig: reduction 3 times w_e auxiliary edge removed}. Say the removed edge is \(e_v e_u\), then the most logical way to map this edge-stabilizer to the original graph is by mapping it to the edge \(uv\). However, as can be seen in \Cref{subfig: reduction 3 times w_e edge removed}, removing this edge does not result in a stable graph.

The above example also proves that edge-stabilizers are not preserved with the second reduction, as if one applies the construction discussed in the second reduction (instead of this latter one), the same conclusion can be made: edge-stabilizer in the unit-capacity auxiliary graph might not map to an edge-stabilizer in the original graph.

All in all, it is possible that with some extra work some reductions turn out to work, but not in a straightforward manner, as (1) the dual values of the vertices can mess up the algorithm by decreasing the value of the matching too much, (2) edge-stabilizers are not preserved, and (3) we cannot use the reduction on unit-weight instances. 
So such directions (if they work) seem to only make the problem (unnecessarily) more complicated. Differently, our way is direct, cleaner and simpler.
\section{Basic Fractional c-Matchings}
\label{appx: basic f-c-matching}

\basicfraccmatching*
\begin{proof}
    (\(\Rightarrow\))
    Let \(x\) be a basic fractional \(\capacity\)-matching, and let \(H\) be a connected component of the graph induced by the edges with fractional value in \(x\).
    First, note that \(H\) contains no even cycle, and no inclusion-wise maximal path with distinct endpoints: Otherwise, let \(D\) be an even cycle or an inclusion-wise maximal path with distinct endpoints in \(H\). Let \(g \in \cC_1 \cup \cC_3\) be the circuit associated to it. Then, \(x + \varepsilon g\) and \(x-\varepsilon g\) are both fractional \(\capacity\)-matchings, for a small value of \(\varepsilon\). However, as \(x\) is a convex combination of \(x+ \varepsilon g\) and \(x-\varepsilon g\), this contradicts that \(x\) is an extreme point.
    Let \(T\) be any spanning tree of \(H\). First, assume there exist two distinct edges \(f_1,f_2 \in E(H) \setminus E(T)\). Then, adding \(f_1\) (resp.\ \(f_2\)) to \(T\) creates an odd cycle \(D_1\) (resp.\ \(D_2\)). These cycles cannot intersect in an edge, otherwise their support would contain an even cycle. Hence, they must be edge disjoint. The cycles can also not intersect in more than one vertex, otherwise their support again contains an even cycle. So, they either intersect at one vertex, or are connected via a path in \(T\). In either case, one can associated to these edges a circuit \(g \in \cC_5\). Then, \(x + \varepsilon g\) and \(x - \varepsilon g\) are both fractional \(\capacity\)-matchings, for a small value of \(\varepsilon\), reaching a contradiction again. These arguments show that there is a unique edge \(f \in E(H) \setminus E(T)\). Since \(H\) cannot contain inclusion-wise maximal paths, it contains at most one vertex of degree 1. If \(H\) contains exactly one such vertex \(u\), then \(u\) and the odd cycle (created by adding \(f\) to \(T\)) are connected via a path. One can associate a circuit \(g \in \cC_4\) to the edges in this cycle and path. Again, considering \(x + \varepsilon g\) and \(x - \varepsilon g\) results in a contradiction. So, \(H\) does not have any vertex of degree 1, which means that the endpoints of \(f\) must be the leaves of \(T\). Hence, \(T\) is a path and \(H\) is a cycle. Necessarily, \(H\) must be odd. Finally, no vertex in \(H\) can be unsaturated: as otherwise we can associate a circuit \(g \in \cC_2\) to \(H\) where the unsaturated vertex \(u\) is the only vertex with \(g(\delta(u)) \neq 0\). Once again, we reach a contradiction by considering \(x + \varepsilon g\) and \(x-\varepsilon g\). Since \(H\) is an odd cycle with saturated vertices, the edges in \(H\) must have value \(\frac12\) in \(x\). In conclusion, if \(x_e \notin \CBra{0,1}\), it must equal \(\frac{1}{2}\) and \(e\) must be part of an odd cycle. Furthermore, these odd cycles are vertex-disjoint and all vertices part of an odd cycle are saturated.

    (\(\Leftarrow\)) Consider a vector \(w\) with \(w_e = 1\) for all edges \(e\) in the support of \(x\), and \(w_e = -1\) for all other edges. Then, \(x\) is the unique optimal solution when maximizing  the function \(w\) over \(\Pfcm\). Hence \(x\) is an extreme point.     
\end{proof}
\section{Maximum-Weight Basic Fractional c-Matching with Minimum Number of Odd Cycles}
\label{appx: min nr odd cycles}

\cite{Koh2020Stabilizing} proposes an algorithm to obtain a maximum-weight basic fractional matching with minimum number of odd cycles (\(\Abs{\oddcyclesx{x}} = \gamma(\graph)\)). We generalize their result to \(\capacity\)-matchings.

Given a graph \((\graphVE,\weight,\capacity)\), reduce it to a unit-capacity graph \((\hat{\graph},\hat{\weight},1)\) using the third reduction given in \Cref{appx: reduction}. Compute a basic maximum-weight fractional matching \(\hat{x}\) in \(\hat{\graph}\) with minimum number of odd cycles (\(\Abs{\oddcyclesx{\hat{x}}} = \gamma(\hat{\graph})\)), using the algorithm from \cite{Koh2020Stabilizing}. Consider the subgraph of \(\hat{\graph}\) corresponding with an edge \(e \in \edgeSet\). Since \(\hat{x}\) is basic and maximum-weight, there is only a few ways this subgraph can look, see \Cref{subfig:possibilitiesSubgraph1,subfig:possibilitiesSubgraph2,subfig:possibilitiesSubgraph3,subfig:possibilitiesSubgraph4,subfig:possibilitiesSubgraph5}.
\begin{figure}[ht]
    \centering
    \begin{subfigure}[t]{.3\textwidth}
        \centering
        \begin{tikzpicture}[
            scale=.5,
            circ/.style={
                circle,
                draw=black,
                inner sep=2.5pt
            },
            every node/.style={scale=0.8}
        ]
            \node[circ,label={above:\(u_i\)}] (a) at (0,0) {};
            \node[circ,label=above:\(e_v\)] (b) at (2,0) {};
            \node[circ,label=above:\(e_u\)] (c) at (4,0) {};
            \node[circ,label=above:\(v_j\)] (d) at (6,0) {};
    
            \draw (a) to (b);
            \draw[ultra thick] (b) to (c);
            \draw (c) to (d);
        \end{tikzpicture}
        \caption{}
        \label{subfig:possibilitiesSubgraph1}
    \end{subfigure}
    \hfill
    \begin{subfigure}[t]{.3\textwidth}
        \centering
        \begin{tikzpicture}[
            scale=.5,
            circ/.style={
                circle,
                draw=black,
                inner sep=2.5pt
            },
            every node/.style={scale=0.8}
        ]
            \node[circ,label={above:\(u_i\)}] (a) at (0,0) {};
            \node[circ,label=above:\(e_v\)] (b) at (2,0) {};
            \node[circ,label=above:\(e_u\)] (c) at (4,0) {};
            \node[circ,label=above:\(v_j\)] (d) at (6,0) {};
    
            \draw[ultra thick] (a) to (b);
            \draw (b) to (c);
            \draw[ultra thick] (c) to (d);
        \end{tikzpicture}
        \caption{}
        \label{subfig:possibilitiesSubgraph2}
    \end{subfigure}
    \hfill
    \begin{subfigure}[t]{.3\textwidth}
        \centering
        \begin{tikzpicture}[
            scale=.5,
            circ/.style={
                circle,
                draw=black,
                inner sep=2.5pt
            },
            every node/.style={scale=0.8}
        ]
            \node[circ,label={above:\(u_i\)}] (a) at (0,0) {};
            \node[circ,label=above:\(e_v\)] (b) at (2,0) {};
            \node[circ,label=above:\(e_u\)] (c) at (4,0) {};
            \node[circ,label=above:\(v_j\)] (d) at (6,0) {};
    
            \draw[dashed] (a) to (b);
            \draw[dashed] (b) to (c);
            \draw[dashed] (c) to (d);
        \end{tikzpicture}
        \caption{}
        \label{subfig:possibilitiesSubgraph3}
    \end{subfigure}
    \begin{subfigure}[t]{.3\textwidth}
        \centering
        \begin{tikzpicture}[
            scale=.5,
            circ/.style={
                circle,
                draw=black,
                inner sep=2.5pt
            },
            every node/.style={scale=0.8}
        ]
            \node[circ,label={above:\(u_i\)}] (a) at (0,0) {};
            \node[circ,label=above:\(e_v\)] (b) at (2,0) {};
            \node[circ,label=above:\(e_u\)] (c) at (4,0) {};
            \node[circ,label=above:\(v_j\)] (d) at (6,0) {};
            \node[label={[text=white]above:\(f_s\)}] (s) at (0,.75) {};
            \node[label={[text=white]below:\(g_t\)}] (t) at (0,-.75) {};
    
            \draw[ultra thick] (a) to (b);
            \draw (b) to (c);
            \draw (c) to (d);
        \end{tikzpicture}
        \caption{}
        \label{subfig:possibilitiesSubgraph4}
    \end{subfigure}
    \hfill
    \begin{subfigure}[t]{.3\textwidth}
        \centering
        \begin{tikzpicture}[
            scale=.5,
            circ/.style={
                circle,
                draw=black,
                inner sep=2.5pt
            },
            every node/.style={scale=0.8}
        ]
            \node[circ,label={above:\(f_s\)}] (s) at (-2,.75) {};
            \node[circ,label={below:\(g_t\)}] (t) at (-2,-.75) {};
            \node[circ,label={above:\(u_i\)}] (a1) at (0,.75) {};
            \node[circ,label={below:\(u_k\)}] (a2) at (0,-.75) {};
            \node[circ,label=above:\(e_v\)] (b) at (2,0) {};
            \node[circ,label=above:\(e_u\)] (c) at (4,0) {};
            \node[circ,label=above:\(v_j\)] (d) at (6,0) {};

            \draw[dashed] (s) to (a1);
            \draw[dashed] (t) to (a2);
            \draw[dashed] (a1) to (b);
            \draw[dashed] (a2) to (b);
            \draw (b) to (c);
        \end{tikzpicture}
        \caption{}
        \label{subfig:possibilitiesSubgraph5}
    \end{subfigure}
    \hfill
    \begin{subfigure}[t]{.3\textwidth}
        \centering
        \begin{tikzpicture}[
            scale=.5,
            circ/.style={
                circle,
                draw=black,
                inner sep=2.5pt
            },
            every node/.style={scale=0.8}
        ]
            \node[circ,label={above:\(f_s\)}] (s) at (-2,.75) {};
            \node[circ,label={below:\(g_t\)}] (t) at (-2,-.75) {};
            \node[circ,label={above:\(u_i\)}] (a1) at (0,.75) {};
            \node[circ,label={below:\(u_k\)}] (a2) at (0,-.75) {};
            \node[circ,label=above:\(e_v\)] (b) at (2,0) {};
            \node[circ,label=above:\(e_u\)] (c) at (4,0) {};
            \node[circ,label=above:\(v_j\)] (d) at (6,0) {};

            \draw[dashed] (s) to (a2);
            \draw[dashed] (t) to (a2);
            \draw[ultra thick] (a1) to (b);
            \draw (a2) to (b);
            \draw (b) to (c);
        \end{tikzpicture}
        \caption{}
        \label{subfig:possibilitiesSubgraph6}
    \end{subfigure}
    \caption{ \Cref{subfig:possibilitiesSubgraph1,subfig:possibilitiesSubgraph2,subfig:possibilitiesSubgraph3,subfig:possibilitiesSubgraph4,subfig:possibilitiesSubgraph5} indicate the possible scenario's for the subgraph of \(\hat{G}\) corresponding with an edge \(e \in E\).
    \Cref{subfig:possibilitiesSubgraph6} shows how the matching in \Cref{subfig:possibilitiesSubgraph5} can be changed without affecting the weight of the matching.
    Normal, dashed and bold edges indicate an \(\hat{x}\) value of \(0\), \(\frac12\) and \(1\), respectively. No edge between \(e_u\) and \(v_j\) indicates the value of the edges between \(e_u\) and the copies of \(v\) are irrelevant. For clarity, we have only drawn the relevant vertices. (\(e=uv\), \(f=us\), \(g=ut\))}
    \label{fig:possibilitiesSubgraph}
\end{figure}
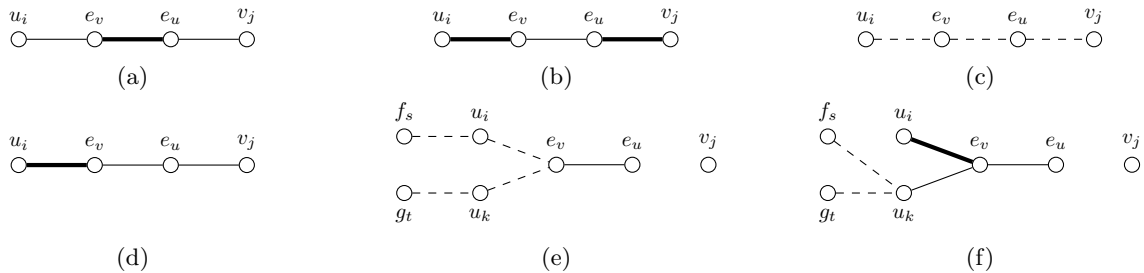

In fact, we can assume without loss of generality we only have the scenario's shown in \Cref{subfig:possibilitiesSubgraph1,subfig:possibilitiesSubgraph2,subfig:possibilitiesSubgraph3}. Indeed, \Cref{subfig:possibilitiesSubgraph4,subfig:possibilitiesSubgraph5} can be transformed into \Cref{subfig:possibilitiesSubgraph1,subfig:possibilitiesSubgraph6}, respectively, without affecting \(\hat{\weight}^\top \hat{x}\) and \(\Abs{\oddcyclesx{\hat{x}}}\). The scenario in \Cref{subfig:possibilitiesSubgraph6}, depending on what is happening between \(e_u\) and the copies of \(v\), corresponds to the scenario in \Cref{subfig:possibilitiesSubgraph2}, \ref{subfig:possibilitiesSubgraph4}, or \ref{subfig:possibilitiesSubgraph5}.

Then \(\hat{x}\) can be translated to a fractional \(\capacity\)-matching \(x\) in \(G\), as follows: set \(x_e=0\) in case of \Cref{subfig:possibilitiesSubgraph1}, \(x_e=1\) in case of \Cref{subfig:possibilitiesSubgraph2}, and \(x_e=\frac12\) in case of \Cref{subfig:possibilitiesSubgraph3}. We have \(\weight^\top x = \hat{\weight}^\top \hat{x} - \weight(\edgeSet)\).
A fractional odd cycle \(\hat{C}\) in \(\hat{x}\) maps to a fractional cycle \(C\) in \(x\) with length \(\Abs{\hat{C}}/3\), which is odd. Suppose there is an unsaturated vertex \(u\) on \(C\), then there is a copy \(u_i\) of \(u\) in \(\hat{G}\) that is exposed. Let \(e=uv\) be an edge on \(C\). We can then change \(\hat{x}\) as follows: set \(\hat{x}_{u_i e_v} = 1\), and alternate round \(\hat{C}\) exposing \(e_v\). One can check using complementary slackness that this gives a new maximum-weight matching in \(\hat{G}\), that is basic and has less fractional odd cycles than \(\hat{x}\), a contradiction. So all vertices on \(C\) are saturated, which means that \(x\) is basic.
It follows that \(\Abs{\oddcyclesx{x}} = \Abs{\oddcyclesx{\hat{x}}}\).

Note that given an \(x\), we can similarly translate it into \(\hat{x}\), such that \(\hat{\weight}^\top \hat{x} = \weight^\top x + \weight(E)\) and \(\Abs{\oddcyclesx{\hat{x}}} = \Abs{\oddcyclesx{x}}\). Hence, \(\nufcG = \nu_f(\hat{\graph}) - \weight(\edgeSet)\) and \(\gamma(\graph) = \gamma(\hat{\graph})\). So, \(x\) is our basic maximum-weight fractional \(\capacity\)-matching in \(\graph\) with \(\Abs{\oddcyclesx{x}} = \gamma(\graph)\).

\bibliographystyle{plain}
\bibliography{bib}

\end{document}